\documentclass[11pt]{article}
\usepackage[utf8]{inputenc}
\usepackage{amsmath,amsfonts,amssymb}
\usepackage{amsthm}
\usepackage{latexsym}
\usepackage[noadjust]{cite}
\usepackage{graphicx}
\usepackage{xcolor}
\usepackage{mathtools}
\usepackage[margin=1in]{geometry}
\usepackage{thm-restate}
\usepackage{enumitem}
\usepackage[linesnumbered,ruled]{algorithm2e}
\usepackage[colorlinks=true, allcolors=blue]{hyperref}
\usepackage[nameinlink,capitalise]{cleveref}
\hypersetup{
    citecolor={violet}
}

\newtheorem{theorem}{Theorem}[section]

\newtheorem{lemma}[theorem]{Lemma}

\newtheorem{claim}[theorem]{Claim}
\theoremstyle{definition}
\newtheorem{definition}[theorem]{Definition}
\newtheorem{remark}[theorem]{Remark}

\newcommand{\N}{\mathbb{N}}
\newcommand{\E}{\mathbb{E}}
\newcommand{\R}{\mathbb{R}}

\newcommand{\rank}{\text{rank}}

\newcommand{\eps}{\epsilon}

\newcommand{\M}{\mathcal{M}}
\newcommand{\I}{\mathcal{I}}

\title{An $\widetilde{O}(n^{3/7})$ Round Parallel Algorithm for Matroid Bases}
\date{\today}

\author{Sanjeev Khanna\thanks{Courant Institute, Warren Weaver Hall, New York University, 251 Mercer Street, New York, NY 10012.  Supported in part by NSF award CCF-2625203 and AFOSR award FA9550-25-1-0107. Email: {\tt sanjeev.khanna@nyu.edu}.} \and Aaron Putterman\thanks{School of Engineering and Applied Sciences, Harvard University, Cambridge, Massachusetts, USA. Supported in part by the Simons Investigator Awards of Madhu Sudan and Salil Vadhan and AFOSR award FA9550-25-1-0112. Email: \texttt{aputterman@g.harvard.edu}.} \and Junkai Song\thanks{Courant Institute, Warren Weaver Hall, New York University, 251 Mercer Street, New York, NY 10012. Supported in part by NSF award CCF-2625203. Email: \texttt{junkaisong@nyu.edu}.}}

\begin{document}

\pagenumbering{gobble}

\maketitle

\begin{abstract}
We study the parallel (adaptive) complexity of the classic problem of finding a basis in an $n$-element matroid, given access via an \emph{independence oracle}. 
In this model, the algorithm may submit polynomially many independence queries in each round, and the central question is: 
how many rounds are necessary and sufficient to find a basis?

Karp, Upfal, and Wigderson (FOCS~1985, JCSS~1988; hereafter KUW) initiated this study, showing that $O(\sqrt{n})$ adaptive rounds suffice for any matroid, and that $\widetilde\Omega(n^{1/3})$ rounds are necessary even for partition matroids. 
This left a substantial gap that persisted for nearly four decades, until Khanna, Putterman, and Song (FOCS~2025; hereafter KPS) achieved $\widetilde O(n^{7/15})$ rounds, the first improvement since~KUW.

In this work, we make another conceptual advance beyond KPS, giving a new algorithm that finds a matroid basis in $\widetilde O(n^{3/7})$ rounds. 
We develop a structural and algorithmic framework that brings a new lens to the analysis of random circuits, moving from reasoning about individual elements to understanding how dependencies span multiple elements simultaneously.
Specifically, our framework introduces three new ideas:

\begin{enumerate}
  \item A new \emph{subset-based decomposition} that provides precise guarantees on how random circuits intersect groups of elements, yet remains computable in few adaptive rounds.
  \item A new method for identifying and removing redundant elements in bulk, based on short circuit witnesses that certify redundancy across large portions of the matroid.
  \item An adaptive early-stopping strategy that uses the evolving structure of the matroid to decide when to contract or delete, preventing wasted rounds.
\end{enumerate}

Each of these contributions, in isolation, already yields meaningful improvements over the round complexity achieved in KPS; their combination enables our main result of $\widetilde O(n^{3/7})$ rounds.

As further consequences, incorporating our improved basis-finding algorithm into known reductions yields an
$\widetilde O(n^{17/21})$-round parallel algorithm for matroid intersection, as well as an 
$\widetilde O(n^{3/7})$-round parallel algorithm for approximate monotone submodular maximization under a matroid constraint.
\end{abstract}

\pagebreak  

\tableofcontents

\pagebreak 

\pagenumbering{arabic}

\section{Introduction} \label{sec:intro}

In this work, we continue the study of designing efficient parallel algorithms for fundamental combinatorial optimization problems. 
Parallel efficiency is typically measured through {\em adaptive round complexity}, the number of rounds of operations/queries required to solve a problem, where each round allows polynomially many operations/queries to be executed in parallel. 
Many foundational problems have been explored under this lens, beginning with the seminal works on maximal independent sets~\cite{Lub86}, spanning trees \cite{KUW85, khanna2025optimal}, and  graph matchings~\cite{Lov79, KUW86, FGT16, ST17}, and extending to more recent advances in submodular function minimization~\cite{BS20, CCK21, CGJS22}, submodular function maximization~\cite{BS18a,BS18b,BRS19a,BRS19b,CQ19a,CQ19b,EN19,ENV19,FMZ19a,FMZ19b,KMZ+19,CFK19,BBS20,LLV20}, and matroid intersection~\cite{GGR22, GT17, Bli22, BT25}.

In this paper, we focus specifically on understanding the parallel complexity of finding a \emph{basis} of a matroid.
Formally, a matroid $\mathcal{M} = (E, \mathcal{I})$ consists of a ground set $E$ of $n$ elements and a collection of independent sets $\mathcal{I} \subseteq 2^E$ satisfying:
\begin{itemize}
\item $\emptyset \in \mathcal{I}$;
\item if $S \in \mathcal{I}$ and $S' \subseteq S$, then $S' \in \mathcal{I}$; and
\item if $S, T \in \mathcal{I}$ and $|S| < |T|$, then there exists $e \in T \setminus S$ such that $S \cup \{ e \} \in \mathcal{I}$.
\end{itemize}

A \emph{basis} is an independent set $S \subseteq E$ that is maximal under inclusion. 
Because matroids capture independence across a wide range of structures, from cycle-freeness in graphs to linear independence in vector spaces, finding a basis is a fundamental combinatorial primitive. 
For example, in graphs, a matroid basis corresponds to a spanning forest; in vector spaces, it coincides with the standard notion of a basis.

Given the broad applicability of matroids, understanding the parallel complexity of basis finding remains a central theoretical challenge. 
Despite decades of work, a substantial gap persists in characterizing how efficiently a matroid basis can be computed in parallel. 
The difficulty stems from the generality of matroids: the number of distinct matroids grows super-exponentially in the size $n$ of the ground set~\cite{BPV15}. This fact points both to the rich combinatorial structure of matroids, and the impossibility of representing them succinctly.
Thus in the general setting considered in this work, it is typical to assume that the access to a matroid is given via an oracle.
Specifically, we will assume the matroid~$\mathcal{M}$ is presented via an \emph{independence oracle}~$\mathrm{Ind}$, which reports membership in~$\mathcal{I}$ for any queried subset~$S \subseteq E$. 
In each adaptive round, the algorithm may make polynomially many independence queries.

This formalization of the problem was first introduced by Karp, Upfal, and Wigderson ~\cite{KUW85, KUW88} (KUW), who asked:

\begin{quote}
\emph{Given independence oracle access to an arbitrary matroid $\mathcal{M}$, how many rounds of polynomially many queries are required to find a basis?}
\end{quote}

KUW~\cite{KUW85} provided the first set of foundational results, showing that $O(\sqrt{n})$ adaptive rounds suffice and that $\widetilde{\Omega}(n^{1/3})$ rounds are necessary even for partition matroids. 
This left a broad gap that persisted for nearly forty years, until Khanna, Putterman, and Song (KPS)~\cite{KPS25} developed a new decomposition-based algorithm achieving $\widetilde{O}(n^{7/15})$ rounds, the first improvement since~KUW. 
Nevertheless, the precise adaptive complexity of matroid basis finding remains unresolved.

\paragraph{Our Contributions.}

We make further progress on this long-standing question by presenting a new parallel algorithm with improved round complexity:

\begin{theorem}\label{thm:3-7roundsintro}
There is an algorithm which for any arbitrary matroid $\mathcal{M}$ on $n$ elements, makes polynomially many independence queries per round and recovers a basis of $\mathcal{M}$ in $\widetilde{O}(n^{3/7})$ rounds with high probability. 
\end{theorem}

KUW established a lower bound of $\widetilde{\Omega}(n^{1/3})$ rounds, so our result narrows the possible exponent to the interval $[\tfrac{1}{3}, \tfrac{3}{7}] \approx [0.333, 0.429]$, a significant improvement over the previous upper bound of $\tfrac{7}{15} \approx 0.467$ from~KPS. 
Our approach introduces a new structural and algorithmic framework that brings a \emph{new analytical lens} to random-circuit structure, moving from reasoning about individual elements to understanding how dependencies spread across groups of elements. 
This shift enables three new algorithmic components: a refined matroid decomposition, a method for bulk identification of redundant elements, and an adaptive meta-strategy for deciding when to contract or delete, which together yield the $\widetilde{O}(n^{3/7})$ bound. We explain these contributions and their context in more detail in \cref{sec:technicaloverview}.

\paragraph{Applications.}

The adaptive complexity model has also been explored in several related matroid optimization problems, where the task of finding a matroid basis or computing the rank of the matroid (equivalently, the size of a basis) often serves as a fundamental subroutine. Consequently, our improved basis-finding algorithm directly yields sharper bounds.

\paragraph{Matroid Intersection.}

In the matroid intersection problem, we are given two matroids $\mathcal{M}_1 = (E, \mathcal{I}_1)$ and $\mathcal{M}_2 = (E, \mathcal{I}_2)$ on the same ground set, and the goal is to find a largest set $S \subseteq E$ such that $S \in \mathcal{I}_1 \cap \mathcal{I}_2$. In the sequential setting, Edmonds~\cite{Edm70,Edm09} gave the first polynomial time algorithm, and subsequent work has progressively improved the running time, culminating in an $\tilde{O}(n\sqrt{r})$ rank-query algorithm~\cite{CLS+19} and an $\tilde{O}(nr^{3/4})$ independence-query algorithm~\cite{Bli21} when the size of the intersection is $r$.

In the parallel setting, Chakrabarty, Chen, and Khanna~\cite{CCK21} established an $\Omega(n^{1/3})$ round lower bound for rank-query algorithms, improving upon the classical $\Omega(n^{1/3})$ lower bound for independence-query algorithms by~\cite{KUW85}. On the upper bound side, Blikstad~\cite{Bli22} presented the first sublinear round algorithms in both the rank-query and independence-query models, later improved by~\cite{BT25}, who gave an $O(n^{2/3})$ round rank-query algorithm, and an $O(n^{5/6})$ round independence-query algorithm.

Our improved basis finding algorithm has immediate implications for this classic problem: indeed, using our improved basis-finding procedure within the work of~\cite{BT25}, we obtain the following round complexity for matroid intersection (proved formally in \cref{sec:matroidIntersection}).
\begin{theorem}\label{thm:matroid-intersection-intro}
There is an $\widetilde{O}(n^{17/21})$-round algorithm that, given any two matroids $\M_1$ and $\M_2$ on the same ground set of $n$ elements, 
makes polynomially many independence queries per round and, with high probability, outputs a maximum common independent set of $\M_1$ and $\M_2$.
\end{theorem}

This improves upon the previous best complexity of $\widetilde{O}(n^{37/45})$ rounds established in~\cite{KPS25}. 

\paragraph{Submodular Function Maximization under a Matroid Constraint.}

In the \emph{submodular function maximization} problem, we are given a monotone submodular function $f: 2^E \to \mathbb{R}_{\ge 0}$ and a matroid $\M = (E, \I)$, and the goal is to find an independent set $S \in \I$ that approximately maximizes $f(S)$. In the sequential setting, a $(1-1/e)$-approximation algorithm has been known in the case where $\M$ is a uniform matroid~\cite{NWF78}, and Vondr\'{a}k~\cite{Von08} extended this guarantee to general matroids. The approximation factor is optimal with polynomially many queries (here queries refer to evaluating the function $f$ on a chosen set $S \subseteq E$)~\cite{NW78,Von13}.

Balkanski and Singer~\cite{BS18a} initiated the study of this problem in the parallel complexity model. When $\M$ is a uniform matroid, they proved that any approximation better than $O(1/\log n)$ requires $\Omega(\log n/\log \log n)$ rounds, and designed a $1/3$-approximation algorithm in $O(\log n)$ rounds. Subsequent works achieved $1-1/e-\epsilon$ approximation algorithms using $O(\text{poly}(\log n,1/\epsilon))$ rounds, both when $\M$ is a uniform matroid~\cite{BRS19a,CQ19a,EN19,FMZ19a}, and when $\M$ is an arbitrary matroid accessed via rank queries~\cite{BRS19b,CQ19b,ENV19}. However, when $\M$ is accessed via independence queries, their algorithms only achieved a round complexity of $\widetilde{O}(\sqrt{n} / \eps^3)$, as their computation relies explicitly on finding a basis\cite{BRS19b}. Later, a separate work by Li, Liu and Vondr\'{a}k~\cite{LLV20} showed that a polynomial dependence on $\epsilon$ in the round complexity is unavoidable.

As in matroid intersection, by combining the algorithm of~\cite{BRS19b} with our improved matroid basis finding algorithm, we obtain immediate improvements in the round complexity. The resulting guarantee is stated below; a formal proof appears in~\cref{sec:matroidIntersection}.

\begin{theorem}\label{thm:submodular-maximization}
For any $\epsilon > 0$, there is an $\widetilde{O}(n^{3/7}\epsilon^{-3})$-round algorithm that, given any matroid $\M=(E,\I)$ on $n$ elements and any monotone submodular function $f : 2^E \to \R_{\ge 0}$, 
makes polynomially many independence queries per round and outputs, with high probability, a $(1 - 1/e - O(\epsilon))$-approximation to the maximum of $f$ under the matroid constraint~$\M$.
\end{theorem}
This improves the previous best bound of $\widetilde{O}(n^{7/15}\epsilon^{-3})$ rounds, achievable via the algorithm of~\cite{KPS25}, though not explicitly stated in their work.

We next provide background on prior work and terminology in \cref{sec:priorWork}, before turning to an overview of our new techniques in \cref{sec:technicaloverview}.

\subsection{Prior Work}\label{sec:priorWork}

We start by recalling some matroid terminology that will be used throughout the introduction and reviewing the approaches of KUW~\cite{KUW85,KUW88} and KPS~\cite{KPS25}, as these provide the key backdrop for our contribution.

\paragraph{Notation and Terminology.}
For a matroid $\M = (E, \I)$, a \textbf{basis} is a set $B \subseteq E$ that is a \emph{maximal independent set}: $B \in \I$, but $B \cup \{e \} \notin \I$ for all $e \in E \setminus B$. A \textbf{circuit} is a set $C \subseteq E$ that is a \emph{minimal dependent set}: $C \notin \I$ but $C \setminus \{e\} \in \I$ for all $e \in C$. 
We write $\mathrm{rank}(S)$ for the size of the largest independent subset of $S \subseteq E$, and define the \textbf{span} of $S$ as 
$\mathrm{span}(S) = \{e \in E : \mathrm{rank}(S) = \mathrm{rank}(S \cup \{e\})\}$. 
That is, $\mathrm{span}(S)$ consists of all elements whose addition to $S$ does not increase its rank.

We will rely on several standard facts about matroids (see, e.g.,~\cite{Oxl11}). 
If $S \subseteq E$ satisfies $\mathrm{rank}(S) = \mathrm{rank}(E)$, then any basis of $S$ is also a basis of $E$. 
Moreover, if $S$ is independent, then there exists a basis of~$\M$ that contains~$S$ by the extension property of matroids. 
This fact motivates the operation of \textbf{contraction}: given $\M = (E, \I)$ and $S \in \I$, the contracted matroid $\M / S$ is defined on the ground set $E \setminus S$ so that $T \subseteq E \setminus S$ is independent in $\M / S$ if and only if $T \cup S$ is independent in~$\M$. 
Intuitively, contraction corresponds to ``committing'' to include $S$ in the eventual basis; if one finds a basis~$T$ of $\M / S$, then $S \cup T$ forms a basis of the original matroid $\M$.

\paragraph{Formal Problem Statement} 
As mentioned above, we study the \emph{parallel complexity} of finding a basis of a matroid. 
Formally, an algorithm~$A$ is given access to a matroid $\M = (E, \I)$ through an \emph{independence oracle}~$\mathrm{Ind}$, which, for any set $S \subseteq E$, returns $\mathbf{1}[S \in \I]$. 
The computation proceeds in \emph{rounds}: in each round, the algorithm may issue up to $\mathrm{poly}(n)$ oracle queries in parallel, where each query is a subset $S \subseteq E$, and the oracle responds whether $S$ is independent. 
Importantly, the queries in the $i$th round are made in \emph{parallel}, meaning that these queries depend only on responses to queries in rounds $1, \dots i-1$ (and on the algorithm's internal randomness). 

The objective is to find a basis of the matroid $\M$ (with high probability over the randomness of the algorithm) in as few rounds of queries as possible. Lastly, our algorithm should work \emph{for all} matroids $\M$: that is, we define the round complexity to be the \emph{maximum over all matroids} $\M$ on $n$ elements, of the number of rounds the algorithm requires to find a basis of $\M$ with high probability.

\paragraph{Overarching Themes.}

With these preliminaries in place, we turn to the high-level intuition behind both KUW and KPS. 
In both works, progress toward finding a matroid basis proceeds through two complementary operations:

\begin{enumerate}
    \item \textbf{Deleting redundant elements:} If one can identify a set $S$ such that every element of $S$ lies in the span of $E \setminus S$, then deleting $S$ does not reduce the matroid’s rank. 
This means $\mathrm{rank}(E) = \mathrm{rank}(E \setminus S)$ and therefore there exists a basis supported entirely on $E \setminus S$. This reduces the search space of the problem; instead of finding a basis over $\M = (E, \I)$, we instead search over the matroid $\M|_{E \setminus S} = (E \setminus S, \I \cap 2^{E \setminus S})$.
    \item \textbf{Contracting on an independent set:} 
    If one can find an independent set $S$, then by the extension property, there exists a basis containing $S$. 
This allows the algorithm to contract on $S$ and focus subsequent queries on the residual matroid $\M / S$.
\end{enumerate}

\paragraph{\cite{KUW85}'s $O(\sqrt{n})$ Round Algorithm.} 
Using these two operations, \cite{KUW85} designed a simple yet powerful algorithm: partition the ground set $E$ into $\sqrt{n}$ groups of size $\sqrt{n}$ each. 
Within each group $S=\{e_1,\dots,e_{\sqrt{n}}\}$, query the independence oracle on all prefixes $\{e_1\}, \{e_1, e_2\}, \ldots, S$. 
Two outcomes are possible:

\begin{itemize}
    \item If any group $S$ is fully independent, we can contract on it, adding at least $\sqrt{n}$ \emph{independent} elements to the basis.
    \item Otherwise, in every group, the first element that introduces dependence is redundant: 
if $\{e_1,\ldots,e_j\}$ is independent but $\{e_1,\ldots,e_{j+1}\}$ is dependent, then $e_{j+1} \in \mathrm{span}(\{e_1,\ldots,e_j\})$. 
Thus at least one element can be deleted per group, removing $\sqrt{n}$ \emph{redundant} elements in total.
\end{itemize}

Hence, in each round, the instance size decreases from $n$ to $n - \sqrt{n}$, and after $O(\sqrt{n})$ rounds, the algorithm outputs a basis. 
To complement this upper bound, \cite{KUW85} also proved an $\widetilde{\Omega}(n^{1/3})$ lower bound, leaving a large gap that persisted for next four decades.

\paragraph{\cite{KPS25}'s $\widetilde{O}(n^{7/15})$ Round Algorithm.}

The algorithm of \cite{KPS25} takes a fundamentally different approach from the $O(\sqrt{n})$-round algorithm of~\cite{KUW85}. 
Instead of insisting on immediate progress in every round, their algorithm deliberately spends several rounds gathering structural information about the matroid, performing what they refer to as a \emph{matroid decomposition}. 
Only once this decomposition is established does the algorithm proceed to delete redundant elements and contract independent sets. 
While this strategy may not yield $\widetilde{\Omega}(n^{8/15})$ progress in each round individually, it guarantees an \emph{average} progress of that order over the full execution. 
This flexibility is key: it allows the algorithm to partition $\M$ into ``well-behaved'' regions that can then be processed efficiently for contraction or deletion.

To formalize this notion of ``well-behaved'', \cite{KPS25} introduces two key parameters:

\begin{enumerate}
    \item \textbf{The $\boldsymbol{\alpha(S)}$ parameter.} 
    For any subset $S \subseteq E$, $\alpha(S)$ is defined as the smallest integer $\ell$ for which a uniformly random $\ell$-subset of $S$ is independent with probability at most $1/2$:
    \[
    \alpha(S) = \min \left\{ \ell \in [n] : \Pr_{T \sim \binom{S}{\ell}} [\mathrm{Ind}(T) = 1] \le \tfrac{1}{2} \right\}.
    \]
    Intuitively, $\alpha(S)$ captures the point at which dependence typically appears when the elements of $S$ are revealed in random order. 
If we fix $S$ and query every prefix under a random permutation $\pi$ of its elements, $\alpha(S)$ corresponds to the \emph{median} prefix length at which dependence first arises.

 \item \textbf{Marginal circuit probabilities.} 
    For each element $i \in S$, the \emph{marginal circuit probability} $p_i$ measures how likely $i$ is to appear in the \emph{first circuit} formed during a random permutation process. 
Specifically, for a random permutation $\pi$ of $S$, let $C_\pi$ be the first circuit\footnote{Formally, if $i$ is the first index such that $\{e_{\pi(1)},\ldots,e_{\pi(i-1)}\}$ is independent but $\{e_{\pi(1)},\ldots,e_{\pi(i)}\}$ is dependent, then the unique circuit contained in this latter set is denoted $C_\pi$.} that appears when elements are added in order $\pi(1),\pi(2),\dots$.
Then $p_i = \Pr_{\pi}[\,i \in C_\pi\,]$. When needed, we may write $p_{i,S}$ to emphasize that the probability is taken with respect to the set $S$.
   
\end{enumerate}

Using these parameters, \cite{KPS25} establishes the following key structural lemma.

\begin{lemma}[Informal; \cite{KPS25}]
\label{lem:informal_decomp}    
There is a decomposition algorithm for $\mathcal{M} = (E, \mathcal{I})$ that uses $\mathrm{poly}(n)$ independence queries per round and, if it terminates after $\gamma$ rounds:
    \begin{enumerate}
        \item Recovers disjoint sets $S_1, \dots, S_{\gamma}$ such that for every $j\in [\gamma]$ and $i \in S_j$, \begin{align}\label{eq:marginalprob}
            p_{i,S_j} = \widetilde{\Omega} \left ( \frac{1}{|S_j|} \right ).
        \end{align} Note here that we use $p_{i, S_j}$ to denote the marginal probability with respect to $S_j$, \emph{not} the parent matroid $\mathcal{M}$. \label{item:informal_decomp_1}
        \item For any $j,k \in [\gamma], j < k$, 
        \begin{align}\label{item:informal_decomp_2}
            \alpha(S_k) = \frac{\alpha(S_j) |S_k|}{|S_j|} + \Omega \left( \sqrt{ \frac{\alpha(S_j) |S_k|}{|S_j|} } \right).
            \end{align}
        \item For each $j \in [\gamma]$, we can recover an independent set of size $\Omega\left( \frac{\alpha(S_j)}{|S_j|}\cdot n\right)$ in $\mathcal{M}\setminus (\bigcup_{i<j} S_i)$.\label{item:informal_decomp_3}
    \end{enumerate}
\end{lemma}

This decomposition yields several key consequences. 
First, the recurrence in~\Cref{item:informal_decomp_2} implies that the number of parts $\gamma$ is at most~$O(n^{1/3})$. 
Second, if for any $S_j$ the ratio $\alpha(S_j)/|S_j|$ is large, item (3) guarantees that the residual matroid contains a large independent set that can be efficiently recovered, allowing rapid progress via contraction.

The complementary case, when $\alpha(S_j)/|S_j| \ll 1$, indicates the presence of many redundant elements in $S_j$. 
\cite{KPS25} shows that these can be identified and deleted in parallel via the following progress lemma.

\begin{lemma}[Informal; \cite{KPS25}] \label{lem:informal_progress}
    Let $S_j$ be a set peeled off in the above decomposition, such that for all $i \in S_j$, $p_{i,S_j} = \widetilde{\Omega} \left( \frac{1}{|S_j|} \right )$, and $\alpha(S_j) \leq |S_j| / \log^2 |S_j|$. Then:
    \begin{enumerate}
        \item There is an $O(\sqrt{|S_j|})$-round algorithm that recovers $\widetilde{\Omega} ( |S_j| )$ redundant elements.
        \label{item:informal_progress_1}
        \item There is also a $1$-round algorithm that recovers
        \[
            \widetilde{\Omega} \left( \min \left( |S_j|, \frac{|S_j|^2}{\alpha(S_j)^2} \right) \right )
        \]
        redundant elements.
        \label{item:informal_progress_2}
    \end{enumerate}
\end{lemma}

Together, these lemmas yield a ``win–win'' structure: when $\alpha(S_j)$ is relatively large, a sizable independent set can be contracted; when $\alpha(S_j)$ is relatively small, many redundant elements can be deleted. 
By grouping the $S_j$’s by size (into $\log n$ geometric buckets), and focusing on the group with the most sets, \cite{KPS25} shows that there always exists a group where one can achieve $\widetilde{\Omega}(n^{8/15})$ average progress per round (where progress refers to both independent elements that are contracted and redundant elements that are deleted), leading to their $\widetilde{O}(n^{7/15})$-round bound for matroid basis finding.

\subsection{Our New Algorithm}\label{sec:technicaloverview}

We introduce three conceptual advances that change how progress is measured and managed in the matroid-basis problem, leading to a sharper framework and improved round complexity over~\cite{KPS25}. We summarize them here before turning to their technical details.

\begin{enumerate}
    \item \textbf{Subset-hitting decomposition.} 
    We strengthen the decomposition guarantee from element-wise control to \emph{subset-level} control: instead of ensuring $p_{i,S} \gtrsim 1/|S|$ for each element~$i$, we require that \emph{for every subset} $T \subseteq S$, the first circuit intersects $T$ with probability $\widetilde{\Omega}(|T|/|S|)$. 
 This subset-hitting property rules out highly correlated pathological behavior, yielding a more uniform distribution of first circuit mass across elements in a peeled set. Surprisingly, a decomposition with this stronger (global) guarantee is still computable within the \emph{same} round complexity as in~\cite{KPS25} using only polynomially many queries, as before.

    \item   \textbf{Short-circuit witnesses for bulk deletion.} 
    Leveraging subset-hitting, we design a one-round deletion primitive that exploits \emph{circuit sizes}. 
    If there is a set $R \subseteq S$ such that for every $x \in S\!\setminus\!R$ we can certify a circuit $C_x$ with $x \in C_x$ and $|C_x \cap (S\!\setminus\!R)| \le \ell$, then we can delete $\Omega(|S\!\setminus\!R|/\ell)$ redundant elements in one round. By carefully instantiating the set $R$, we show a new win-win paradigm for deleting redundant elements: namely we can either (a) delete lots of elements by finding short circuits that they participate in, or (b) \emph{re-use} a redundant element finding algorithm from \cite{KPS25} but now with \emph{much better guarantees}!
    
  This win–win paradigm relies crucially on subset-hitting; without it (e.g., under the KPS element-wise guarantee), the improved bounds do not materialize.

  \item \textbf{Adaptive early-stopping via evolving $\boldsymbol{\alpha}$-profiles.} 
    Lastly, we introduce a much improved strategy for deciding \emph{when} to contract on independent sets versus when to delete redundant elements. This is based on a new quantitative analysis of how the $\alpha$-parameters of the peeled pieces evolve across the decomposition. 
    This yields an \emph{early-stopping} rule that avoids investing rounds when the average progress target is already met, improving amortized progress.

\end{enumerate}

These new ingredients can be applied in a \emph{modular} yet \emph{complementary} way. 
Using only (1) and (2) with the original~\cite{KPS25} global analysis yields an $\widetilde{O}(n^{0.44})$-round algorithm. 
Using only (3), the adaptive early-stopping analysis along with~\cite{KPS25} primitives, gives an $\widetilde{O}(n^{4/9})$ bound with a substantially simpler proof (see \cref{sec:appendix49}). 
When combined, the three components reinforce one another, though the resulting analysis becomes significantly more delicate. It is this combination that ultimately enables our $\widetilde{O}(n^{3/7})$-round algorithm. 
We now describe each ingredient and how they interact in more detail.

\subsubsection{Intuition for Improved Redundant Element Recovery}

We now elaborate on our new algorithm for identifying redundant elements (the second advance described above). The key starting point is that the structure of \emph{circuits}, that is, the minimal dependent sets, can be exploited much more deeply than before.

Suppose we are trying to find a basis of a matroid $\M = (E, \I)$, and that we have peeled off a set $S_j$ according to \cref{lem:informal_decomp}. For clarity, assume that the marginal circuit probabilities satisfy $p_i = \widetilde{\Omega}(1/|S_j|)$, as in the weaker decomposition of \cref{lem:informal_decomp}. Intuitively, this means that under a random permutation of the elements, the first circuit encountered contains each $i \in S_j$ with probability $\widetilde{\Omega}(1/|S_j|)$—that is, a nontrivial fraction of these random circuits involve~$i$.

Now suppose that for every element $i \in S_j$, we can find a circuit $C_i$ containing $i$ of size at most~$\ell$. In this ideal case, a simple greedy argument shows that at least $|S_j|/\ell$ elements are redundant. Indeed, we may process the elements of $S_j$ in an arbitrary order: for each $i$, keep all other elements of $C_i \setminus \{i\}$ (at most $\ell-1$ of them), since $i \in \mathrm{span}(C_i \setminus \{i\})$. We delete $i$, commit to keeping these witnesses, and repeat. Each deletion costs at most $(\ell-1)$ commitments, so we can delete $\Omega(|S_j|/\ell)$ elements in total.

Of course, not every element will participate in such a small circuit. This motivates examining the \emph{distribution} of circuit sizes. Within our experiment of randomly sampling permutations and identifying the first emerging circuit, our first observation is that the \emph{expected} circuit size equals the sum of the marginal circuit probabilities:
\[
\sum_{i \in S_j} p_i = \mathbb{E}_{\pi}[|C_\pi|],
\]
since $p_i = \Pr[i \in C_{\pi}] = \E_{\pi}[\mathbf{1}[i \in C_{\pi}]]$, and summing over all elements gives the expected size of $C_\pi$.

This identity suggests a natural dichotomy for progress:
if $\sum_i p_i$ is large, many elements appear frequently in the first circuit, and existing methods (such as those in~\cite{KPS25}) can delete a large number of redundant elements.
Conversely, if the average circuit size is small, one might hope that many elements lie in short circuits, opening the door to fast progress through our short-circuit deletion scheme.

However, this intuition requires care: a small \emph{average} circuit size does not guarantee that \emph{most} elements actually belong to small circuits.
A small subset of elements might have large $p_i$ values and appear disproportionately often in small circuits, while the majority of elements participate only in rare but very large circuits. 
In such cases, the algorithm could make little progress on most elements, even when the average circuit size is small.

To illustrate, consider a set $S_j$ consisting of two disjoint uniform matroids, $P_1$ and $P_2$.
Let $P_1$ contain $\sqrt{n}$ elements with rank $\sqrt{n}/2$ (meaning any subset of $\leq \sqrt{n}/2$ elements is independent), and $P_2$ contain $n - \sqrt{n}$ elements with rank $(n - \sqrt{n})/2 + \eta$, for a parameter $\eta$.
Under a random permutation of all elements, when $\eta=0$ the first circuit $C_\pi$ is equally likely to arise in either $P_1$ or $P_2$. 
As $\eta$ increases, circuits become increasingly unlikely to form in $P_2$ first. 
Choosing $\eta \approx \sqrt{n\log n}$ ensures that circuits appear in $P_2$ only a $1/n$ fraction of the time. 
Then:

\begin{enumerate}
    \item The expected circuit size is $\approx \sqrt{n}/2$, since with probability $1 - 1/n$, the first dependence occurs in $P_1$.
    \item Almost all elements are in $P_2$, and they participate only in circuits of size $\ge (n - \sqrt{n})/2$.
\end{enumerate}

Thus, even though the \emph{expected circuit size} is small, there is no way to make better progress than \cite{KPS25}. The main reason is that many elements with small marginal circuit probabilities often co-occur in the same circuits, creating \emph{correlated low-probability clusters}. To overcome this, our new decomposition explicitly prevents such correlations from emerging, ensuring that elements with small marginal probabilities do not systematically appear together.

\subsubsection{An Improved Decomposition}\label{sec:hittingset}

The decomposition guarantee achieved in~\cite{KPS25} (see Lemma~\ref{lem:informal_decomp}) only ensures that for each individual element $i\in S_j$, we have that $p_{i,S_j}  = \Pr_\pi[i\in C_\pi] = \widetilde\Omega\bigl(1/|S_j|\bigr)$.
While this prevents any single element from being too rare, it still admits highly correlated configurations in which many low‐marginal elements systematically co‐occur in the same circuits (as in the preceding example).

Formally, the issue is that although each element may satisfy $p_{i,S_j}\gtrsim 1/n$, there may exist a \emph{large subset} $T$ (e.g., $P_2$ in the example above) with $\Pr_\pi[C_\pi\cap T\neq\emptyset]\approx 1/n$. This occurs because elements in $T$ are highly correlated: either none appear in the first circuit, or $\Omega(|T|)$ of them do.

Our first step is therefore to rule out such set‐level correlations, not just rare elements.
To eliminate this obstacle, we strengthen the decomposition to enforce the following \emph{subset‐hitting property}:
\[
  \forall\,T\subseteq S_j,\quad
  p_{T,S_j}
  = \Pr_\pi[T\cap C_\pi\neq\emptyset]
  = \widetilde\Omega\bigl(|T|/|S_j|\bigr).
\]
That is, \emph{every} subset $T$ of elements is hit by the first circuit with probability
(up to polylogarithmic factors) proportional to its size.
In the example above, $P_2$ would violate this condition, since the property demands a $\widetilde{\Omega}(1)$ intersection probability rather than the $\Theta(1/n)$ that arises from correlation.

Although this requirement is strictly stronger than element-wise marginals (it constrains exponentially many subsets), we show that such parts \emph{do} exist and can be computed within the same round and query bounds as in~\cite{KPS25}. Moreover, our decomposition continues to satisfy all guarantees of Lemma~\ref{lem:informal_decomp} while adding the new subset‐hitting property.
Operationally, this property prevents mass from concentrating on a few large, highly correlated circuits and ensures that circuit mass is distributed in a stable manner across each peeled part, allowing for more powerful routines for deleting redundant elements.

\subsubsection{Core vs. Non-core Elements}\label{sec:introcore}

With this strengthened decomposition theorem established, it still remains to show how we can use the decomposition to explicitly make progress by deleting redundant elements.
The next step is to stratify elements by their marginal probabilities. Intuitively, elements with small marginals are conducive for finding short‐circuit witnesses (thus enabling many deletions), whereas elements with large marginals can be handled by a deletion procedure from \cite{KPS25}.

To this end, we partition the elements of each peeled set $S_j$ into two groups:
\begin{itemize}
  \item[(a)] {\em Core elements:} elements \(i\in S_j\) with $p_{i,S_j}>\frac{\alpha(S_j)^2}{\lvert S_j\rvert^2}$,
  so that \(p_{i,S_j}\cdot(|S_j|^2/\alpha(S_j)^2)>1\), and
  \item[(b)] {\em Non-core elements:} the remaining elements, each with $p_{i,S_j}\le\frac{\alpha(S_j)^2}{\lvert S_j\rvert^2}$.
\end{itemize}

We will denote by $\mathrm{CORE}$, the set of core elements. The choice of the threshold $\frac{\alpha(S_j)^2}{\lvert S_j\rvert^2}$ in defining core vs. non-core elements is deliberate: a lemma (implicitly) established in~\cite{KPS25} shows that one can delete 

\begin{align}\label{eq:progressKPS25deletion}
    \widetilde\Omega\Bigl(\sum_{i\in S_j}\min\{1,\,p_{i,S_j}\cdot(|S_j|^2/\alpha(S_j)^2)\}\Bigr)
\end{align}
redundant elements using a single round of queries.

Restricting \eqref{eq:progressKPS25deletion} to the core immediately yields
$\widetilde\Omega\bigl(\lvert\mathrm{CORE}\rvert\bigr)$ deletions in one round (the second term in the min is $1$ by definition of core). Hence if $\mathrm{CORE}$ is a constant fraction of $S_j$, we delete $\widetilde\Omega(|S_j|)$ elements in one round.

We therefore focus on the complementary case where \(\lvert\mathrm{CORE}\rvert\) is small.
In this event, let $\ell \;=\;\sum_{\,i\in S_j\setminus\mathrm{CORE}}p_{i,S_j}$ denote the sum of all the marginal probabilities of the non-core elements.
Applying \cref{eq:progressKPS25deletion} to these non-core elements implies we can delete
$\widetilde\Omega\Bigl(\tfrac{|S_j|^2}{\alpha(S_j)^2}\,\ell\Bigr)$
elements in a single round.  Thus if $\ell$ is large, we once again ensure many deletions. 

The remaining case is when $\ell$ is small. In this regime we \emph{prove} (using the subset-hitting property) that most non-core elements admit short-circuit witnesses of size $\widetilde{O}(\ell)$ (or more formally, admit circuits who only have a small number of \emph{non-core} elements), which allows us to delete $\widetilde{\Omega}(\lvert S_j\rvert/\ell)$ redundant elements in one round via a modification of the aforementioned short-circuit deletion procedure. We establish this formally below.

\paragraph{Finding redundant elements when non-core circuit mass is small.}

As before, we let $\ell$ denote the total marginal probability mass of the non-core elements: $\ell =\sum_{\,i\in S_j\setminus\mathrm{CORE}}p_{i,S_j}$.
Now, observe that this quantity equals the expected number of non‐core elements in the first circuit:
\begin{align}\label{eq:nonCoreExpectedSize}
\ell =\sum_{\,i\in S_j\setminus\mathrm{CORE}}p_{i,S_j} = \E_{\pi}[|(S_j\setminus\mathrm{CORE}) \cap C_{\pi}|].
\end{align}
This implies that, in expectation, the number of elements in $S_j\setminus\mathrm{CORE}$ in each circuit that appears is bounded by $\ell$.
The question is whether this expectation can be algorithmically leveraged to certify many short circuits for non‐core elements.

It turns out that the subset‐hitting guarantee of our new decomposition yields an affirmative answer.
Informally, sample many (a large polynomial number) of random permutations $\pi$ over $S_j$, and find the first circuits $C_{\pi}$ that form.
Suppose there is a large set $T \subseteq S_j \setminus \mathrm{CORE}$ (say, $|T| = \Omega(|S_j|)$) such that for every $x \in T$ we \emph{fail} to find a short circuit containing $x$,
namely, no circuit $C_x$ with $x\in C_x$ and $|C_x\cap(S_j\setminus\mathrm{CORE})|\le \ell\log n$.
Then every time $C_{\pi}$ intersects $T$ we must have $|C_{\pi} \cap (S_j\setminus\mathrm{CORE})| \geq \ell \log n$.
By the subset‐hitting property, $\Pr_\pi[C_\pi\cap T\neq\emptyset]\gtrsim |T|/|S_j|=\Omega(1)$.
Hence
\[
\E_{\pi}\bigl[\;|C_{\pi} \cap (S_j\setminus\mathrm{CORE})|\;\bigr] \;\ge\; \Omega(1)\cdot \ell\log n \;>\; \ell,
\]
contradicting \eqref{eq:nonCoreExpectedSize}.
Therefore $|T|=o(|S_j|)$: for \emph{most} non‐core elements we do find short circuits, enabling one‐round deletion of $\widetilde\Omega(|S_j|/\ell)$ elements.

To summarize, if the set of core elements forms only a small fraction of $S_j$, the analysis above yields a one-round win–win: either (a) the~\cite{KPS25} deletion bound gives $\widetilde\Omega\left(\tfrac{|S_j|^2}{\alpha(S_j)^2} \cdot \ell\right)$ deletions, or (b) short-circuit witnesses allow $\widetilde\Omega(|S_j|/\ell)$ deletions.
Balancing these outcomes gives an unconditional bound of $\widetilde\Omega\left(\tfrac{|S_j|^{3/2}}{\alpha(S_j)}\right)$ deletions, strictly improving upon~\cite{KPS25}’s $\widetilde\Omega\left(\tfrac{|S_j|^2}{\alpha(S_j)^2}\right)$.
We next show how this stronger local progress integrates into the global round analysis.

\subsubsection{Global Analysis by Leveraging Evolving $\alpha$-Values}

In the preceding discussion, we introduced a new subroutine for making progress on each set $S_j$ produced by our decomposition.
To summarize, we now have two primary ways of making progress:
\begin{itemize}
    \item Contraction: contracting an independent set of size $\Omega\left(\frac{\alpha(S_j)}{|S_j|}\cdot n\right)$ (\cref{lem:informal_decomp}, \cref{item:informal_decomp_3}).
    \item Deletion: deleting $\tilde \Omega\left(\frac{|S_j|^{3/2}}{\alpha(S_j)}\right)$ redundant elements, improving on the $\tilde \Omega\left(\frac{|S_j|^2}{\alpha(S_j)^2}\right)$ bound from \cite{KPS25} (\cref{lem:informal_progress}, \cref{item:informal_progress_2}). (For clarity we omit the $\min\{\cdot,|S_j|\}$ truncation since it does not affect asymptotics here.)
\end{itemize}

The key parameter that controls the efficiency of these routines is $\alpha(S_j)$, which governs both the likelihood of quickly finding a large independent set or deleting in parallel a large number of redundant elements.
To achieve an overall round complexity of $\widetilde{O}(n^{3/7})$, we need to ensure that on {\em average}, each round leads to contraction or deletion of $\widetilde{\Omega}(n^{4/7})$ elements.
We set $f:=n^{4/7}$ as this target average progress per round. While some rounds may fall short (depending on $\alpha(S_j)$), we prove that a prolonged small‐progress regime is impossible; the process must transition to a regime that achieves amortized progress at least~$f$ when averaged over all sets peeled off in the decomposition.

\paragraph{An adaptive strategy.}
In \cite{KPS25}, the matroid decomposition algorithm is always executed until the total size of the peeled off parts reaches at least $n/2$ elements, with each peeling step costing one round.
Our adaptive strategy is to terminate \emph{as soon as} the average progress meets $f$.
In particular, when the algorithm peels off the $j$th part $S_j$, it can terminate immediately if:
\[
\sum_{j'\leq j} \frac{|S_{j'}|^{3/2}}{\alpha(S_{j'})}\geq j\cdot f \text{ (deletion)} \quad \text{or} \quad \frac{\alpha(S_j)}{|S_j|}\cdot n\geq j\cdot f  \text{ (contraction)},
\]
At that point, the progress guarantees imply that the average progress per round is already $\tilde \Omega(f)$.

Intuitively, small $\alpha(S_j)$ favors deletion; large $\alpha(S_j)$ favors contraction. Because $\alpha$ increases along the decomposition (for comparable part sizes), there is a point at which contraction dominates. Continuing to peel beyond this point yields diminishing average progress (the independent set size grows only by a constant factor, while rounds accrue), so we stop as soon as we achieve average progress~$f$.

\paragraph{Improved global analysis.}
The potentially problematic scenario is when the decomposition produces a sequence of peeled sets $S_1, \dots ,S_{\gamma}$ where the average progress per round is always much smaller than $f$.
To rule this out, we choose $f$ as large as possible subject to the condition that no such sequence exists.
Formally, we must forbid sets $S_1,\dots,S_\gamma$ satisfying:
\begin{enumerate}[label=(\alph*)]
    \item $\sum_{j=1}^\gamma |S_j|\leq n$.
    \label{item:prop_1}
    \item $\alpha(S_j)$'s satisfy the recurrence given in \cref{lem:informal_decomp}, \cref{item:informal_decomp_2}.
    \label{item:prop_2}
    \item For every $j\in [\gamma]$, $\sum_{j'\leq j} \frac{|S_{j'}|^{3/2}}{\alpha(S_{j'})}< j\cdot f $ and $\frac{|S_j|}{\alpha(S_j)}\cdot n< j\cdot f$.
    \label{item:prop_3}
\end{enumerate}

The benefit of our adaptive strategy is that the constraints in \cref{item:prop_3} must hold for all prefixes $j\in[\gamma]$, since the decomposition algorithm will terminate as soon as the average progress reaches $f$, not just for the final index $j=\gamma$ as in \cite{KPS25}.
This prefix condition is \emph{strictly stronger} than the terminal condition in~\cite{KPS25}, and permits a larger feasible~$f$.

A surprising consequence is that even without the new decomposition theorem and $\widetilde \Omega\left(\frac{|S_j|^{3/2}}{\alpha(S_j)}\right)$ deletion subroutine, using only the existing $\widetilde \Omega \left( \frac{|S_j|^2}{\alpha(S_j)^2} \right)$ deletion algorithm from \cite{KPS25} (\cref{lem:informal_progress}, \cref{item:informal_progress_2}), the adaptive strategy alone improves the round complexity to $\tilde O(n^{4/9})$.
We give a simple self-contained proof of this fact in \cref{sec:appendix49}, avoiding the intricate case analysis of~\cite{KPS25}.

\paragraph{Incorporating the stronger deletion subroutine.}
Once the new $\widetilde \Omega\left(\tfrac{|S_j|^{3/2}}{\alpha(S_j)}\right)$ deletion routine is used, the analysis becomes substantially more delicate.
In obtaining the $\tilde O(n^{4/9})$ round result, homogeneity in $|S_j|$ and $\alpha(S_j)$ allows reduction to the single ratio $\alpha(S_j)/|S_j|$. The bound $\tfrac{|S_j|^{3/2}}{\alpha(S_j)}$ breaks this homogeneity, so we must track both parameters explicitly, leading to a more complex analysis.

To analyze it, we group the peeled sets $S_1, \dots ,S_{\gamma}$ by size into $\log n$ categories where each category contains sets whose sizes differ by a factor of at most $2$.
We then focus on the \emph{dominant category}, i.e., the size range currently containing the most such sets.
We keep track of when the dominant category changes, and what the new dominant category becomes.
Let $a_j$ denote the number of rounds spent between the $(j-1)$st and $j$th change of dominant category, and let $b_j$ denote the representative size of the sets in the dominant category during this interval.
Using \cref{item:prop_2} and \cref{item:prop_3}, we are able to derive a tight recursive relationship on how these parameters evolve.
In particular, we show in \cref{sec:analysisProg} that for any $j\geq 1$,
\[
b_j = \frac{f^{7 \cdot 2^j - 6}}{n^{4 \cdot 2^j - 4}} \cdot \log^{O(2^j)}(n) , \quad \quad  a_j = \widetilde{O}\left ( \frac{f \cdot b_j}{n}\right ).
\]
Importantly, if we select $f$ to be slightly smaller than $n^{4/7}$, then in the above recursive formulation, these bounds on $b_j$, and hence also $a_j$, will be strongly decaying.
This immediately implies that our decomposition must terminate in a small number of rounds, which is only possible if either (1) the average progress exceeds $f$, or (2) we have peeled off $\Omega(n)$ elements in our decomposition.
Clearly, if case (1) occurs, we have succeeded, and so the only remaining case is (2).
The final piece of our argument shows that, after recovering $\Omega(n)$ elements in our decomposition, there is always a way to recover $\widetilde{\Omega}(n^{3/4})$ redundant or independent elements. Dividing this by our bound on the number of rounds invested so far (using the recurrence relation above) then guarantees our average progress is $\tilde \Omega(f)$.

A more detailed analysis and discussion of this argument appears in \cref{sec:analysisProg}.

\paragraph{Organization.}
The remainder of the paper is organized as follows.
\Cref{sec:newDecomp} introduces our new decomposition algorithm, which guarantees the stronger ``subset‐hitting'' property described above. 
\Cref{sec:circuitProbs} formalizes the tradeoff between circuit sizes and marginal probabilities, yielding our improved subroutine that recovers $\widetilde{\Omega}\!\left(\tfrac{|S_j|^{3/2}}{\alpha(S_j)}\right)$ redundant elements from each part $S_j$ in a single round. 
Finally, \Cref{sec:analysisProg} leverages this refined subroutine to establish the global $\widetilde{O}(n^{3/7})$ round bound through a careful amortized analysis. 
For completeness, several proofs and generalizations of arguments from~\cite{KPS25} are included in the appendix.

\section{Preliminaries}
\begin{definition}[Matroids]
A \emph{matroid} $\M=(E,\I)$ is a pair where $E$ is a finite ground set and $\I \subseteq 2^{E}$ is a collection of independent sets with the following properties: (i) $\emptyset \in \I$ (non-triviality), (ii) for every $S\in \I$ and $S'\subset S$, $S'\in \I$ (downward-closedness), and (iii) for every $S,S'\in \I$ and $|S'|<|S|$, there exists some $x\in S\setminus S'$ such that $S+x\in \I$ (exchange property).
\end{definition}

\begin{definition}[Independent Sets, Circuits, Bases]
For a matroid $\M=(E,\I)$, we say a set $S\subseteq E$ is \emph{independent} if $S\in \I$ and \emph{dependent} otherwise. We call a set $B$ a \emph{basis} if it is a maximal independent set, i.e. for any $x\notin B$, $B+x\notin \I$. We call a set $C$ a \emph{circuit} if it is a minimal dependent set, i.e. for any $x\in C$, $C-x\in \I$.
\end{definition}

\begin{definition}[Rank]
For a matroid $\M=(E,\I)$, we define the \emph{rank} of $\M$ as $\rank(\M) = \max_{S\in \I} |S|$. Further, for any $S\subseteq E$, we define $\rank_{\M}(S) = \max_{T\subseteq S,T\in \I} |T|$. The rank function of a matroid is submodular.
\end{definition}

\begin{definition}[Span]
In a matroid $\M=(E,\I)$, we define $\text{span}(S)$ as
\[
\text{span}(S) = \{x\in E \mid \text{rank}(S\cup \{x\})=\text{rank}(S)\}.
\]
\end{definition}

\begin{definition}[Restriction, Contraction]
Let $\M=(E,\I)$ be a matroid and $S\subseteq E$. We write $\M|_S$ for the restriction of $\M$ to $S$, and we use $M - S$ to mean the restriction of $M$ to $E \setminus S$, and $\M/S$ for the contraction of $\M$ by $S$, whose rank function is $\rank_{\M/S}(T)=\rank_{\M}(S\cup T)-\rank_\M (T)$ for any $T\subseteq E\setminus S$.
\end{definition}

\begin{definition}[Permutation-Induced Circuit]
    Let $\M=(E,\I)$ be a matroid, and let $\pi$ be an arbitrary permutation over $E$. We let $C_{\pi}$ denote the unique first circuit which appears when adding elements in the order of $\pi$. I.e., if $\{e_{\pi(1)}, \dots e_{\pi(j)}\} \in \I$, but $\{e_{\pi(1)}, \dots e_{\pi(j)}, e_{\pi(j+1)}\} \notin \I$, we let $C_{\pi}$ denote the unique circuit in $\{e_{\pi(1)}, \dots e_{\pi(j)}, e_{\pi(j+1)}\}$.
\end{definition}

\section{A New Decomposition Algorithm}\label{sec:newDecomp}

In this section, we introduce the key notion of globally optimal sets and present our new decomposition and its implications. \Cref{subsec:globally_optimal} introduces globally optimal sets and how they can be efficiently identified using the independence oracle, and in \Cref{subsec:iterative_peeling}, we present our iterative decomposition process and highlight properties of the evolving $\alpha()$ values during the decomposition.

\subsection{Globally-Optimal Sets}
\label{subsec:globally_optimal}

To start, we present a modification of the decomposition procedure presented in \cite{KPS25}. Whereas this previous work constructed \emph{greedily-optimal} sets, here we instead introduce \emph{globally-optimal} sets. Before defining these notions, we introduce some key parameters that we will utilize.

\begin{definition}
For a matroid $\M=(E,\I)$ and $S\subseteq E$, we let $
    \alpha(S)$ denote the median number of elements sampled from the $S$ before a circuit forms. I.e., 
    \[
    \alpha(S) = \min \left \{k \in \N: \Pr_{T\sim {S\choose k}} [\mathrm{Ind}(T)=1]\leq \frac{1}{2} \right \}.
    \]
\end{definition}

\begin{definition}
    Let $\M=(E,\I)$ be a matroid over $n$ elements. For $i \in E$, we say that
    \[
    p_{i, \M} = \Pr_{\pi}[i \in C_{\pi}],
    \]
    where $C_{\pi}$ is the unique first circuit that appears when adding elements in the order of $\pi$.
    More generally, for an arbitrary set $T \subseteq E$, we define the \textbf{hitting probability} of the set $T$ as 
    \[
    p_{T, \M} =  \Pr_{\pi}[T \cap C_{\pi} \neq \emptyset].
    \]
\end{definition}

\begin{definition}
    Let $\M=(E,\I)$ be a matroid over $n$ elements. For $T \subseteq E$, we define $T$'s \textbf{circuit mass in $\M$} to be 
    \[
    q_{T,\M} = \Pr_{\pi}[C_{\pi} \subseteq T].
    \]
    Often we will denote this by $q_T$ when the parent matroid $\M$ is clear by context. 
\end{definition}

Before defining globally-optimal and greedily-optimal sets, we first recall that in a single round (and using only polynomially many queries), we can estimate $q_S$ to small error \emph{for every} $S$:

\begin{claim}[Claim 4.4 of \cite{KPS25}]\label{clm:boundedErrorEstimation}
    There is a one round algorithm using polynomially many independence queries, which for a matroid $\M=(E,\I)$ and every subset $S \subseteq E$, yields $\hat{q}_S$ such that
    \[
    |\widehat{q}_S - q_S| \leq \frac{1}{n^2},
    \]
    with probability $1 - 2^{-\Omega(n)}$.
\end{claim}

With this, we can now introduce the notion of a \emph{greedily-optimal} set:

\begin{definition}[Definition 4.9, Claim 4.11 of \cite{KPS25}]\label{def:greedilyOpt}
For a matroid $\M=(E,\I)$, we say a set $S\subseteq E$ is \emph{greedily-optimal} if
\begin{enumerate}
    \item $q_{S}\geq 1-2^{-20}$
    \item For every $x\in S$, \[
p_{x,\M|_S} \geq \frac{1}{2^{21}|S|\log (n)}.
\]
\end{enumerate}
\end{definition}

As discussed in \cref{sec:hittingset}, we strengthen the definition by requiring that the second property holds for arbitrary sets $T\subseteq S$, in the sense that including multiple elements \emph{scales} the hitting probability, rather than just holding individually for each element $x\in S$:

\begin{definition}\label{def:globallyOpt}
     Let $\M=(E,\I)$ be a matroid over $n$ elements. We say that a set $S \subseteq E$ is \textbf{globally optimal} if:
     \begin{enumerate}
         \item $q_{S, \M} \geq 1 - 2^{-20}$.
         \item $\forall T \subseteq S$, it is the case that $p_{T, \M|_S} \geq \frac{|T|}{2^{21} |S| \log(n)}$. 
     \end{enumerate}
     Note that here that $2^{20}$ and $2^{21}$ are just sufficiently large constants to ensure the probabilistic arguments go through.
\end{definition}

\begin{remark}\label{rmk:globalisgreedy}
    Observe that globally-optimal sets (\cref{def:globallyOpt}) are \emph{also} greedily-optimal sets as defined in \cref{def:greedilyOpt}, as we can consider the singleton sets where $T = \{x\}$ for elements $x \in S$. Importantly, this means that going forward, all properties that \cite{KPS25} established for greedily-optimal sets also \emph{automatically} hold for globally-optimal sets.
\end{remark}

Ultimately, our goal is to have an algorithm for \emph{decomposing} the matroid $\M$ into a sequence of sets which are all \emph{globally optimal}, while still maintaining guarantees on how the $\alpha$-values of these sets grow:

\begin{lemma}\label{lem:alphaGrowthFirstStatement}
    Let $\mathcal{M}$ be a matroid. There is a decomposition algorithm, making polynomially many queries, such that if the algorithm runs for $k$ rounds, it recovers sets $S_1, \dots S_{k}$, where every $S_i: i \in [k]$ is \emph{globally-optimal} with respect to $\M - S_1 - \dots S_{i-1}$. Further, for any $\ell\in [\log n]$, if we let $T = \{ i \in [k]: |S_i| \in [2^{\ell}, 2^{\ell+1} -1]\}$, $\gamma = |T|$, and let $a_1, \dots a_{\gamma}$ denote the indices in $T$, then with probability $1 - 2^{- \Omega(n)}$ it must be the case that:
    \begin{enumerate}
        \item $\alpha(S_{a_i}) = \Omega(i^2) \text{ for every }i\in[\gamma]$.
        \item $\gamma = O\left(\sqrt{2^{\ell}}\right).$
        \item $k=O(n^{1/3})$.
    \end{enumerate}
Additionally, we have the property that for any $i < j \in [k]$, $\frac{\alpha(S_j)}{|S_j|} = \Omega \left(\frac{\alpha(S_i)}{|S_i|}\right).$
\end{lemma}

Now, before proving this lemma, we require a few building blocks. To start, we show in the rest of this subsection that a globally optimal set can be constructed by a simple algorithm: we start by setting $S$ to be the ground set $E$, and continue to remove sets $T$ from $S$ that do not alter the probability mass of $\hat q_S$ by too much:

\begin{algorithm}[H]
\caption{GloballyOptimalConstructor$(\M=(E,\I))$}\label{alg:globalOptimal}
    Let $S = E$. \\
    Use \cref{clm:boundedErrorEstimation} to calculate $\hat{q}_T$ for every $T \subseteq S$. \\
    \While{$\exists T \subseteq S$ such that $\hat q_S - \hat q_{S\setminus T} \le \frac{|T|}{2^{20}|S|\log (n)}$} {
        $S\gets S\setminus T$
    }
    \Return{$S$}
\end{algorithm}

Immediately, we have the following claim:

\begin{claim}\label{clm:hittingProbability}
    Suppose we invoke \cref{alg:globalOptimal} on a matroid $\M=(E,\I)$, yielding a set $S$. Then, for every set $T \subseteq S$, it must be the case that 
    \[
    p_{T, \M|_S} \geq \frac{|T|}{2^{21} |S| \log(n)}.
    \]
\end{claim}

\begin{proof}
    Whenever a set $S$ is returned, it must be the case that $\forall T \subseteq S$ that 
    \[
    \hat{q}_{S} - \hat{q}_{S \setminus T} \geq \frac{|T|}{2^{20}|S|\log(n)}.
    \]
    In particular, because $|\widehat{q}_S - q_S| \leq \frac{1}{n^2}$ (with high probability), we also know that 
    \[
    q_S - q_{S \setminus T} \geq \frac{|T|}{2^{20}|S|\log(n)} - \frac{2}{n^2} \geq \frac{|T|}{2^{21} |S| \log(n)}.
    \]
    Next, observe that 
    \[
    q_S - q_{S \setminus T} = \Pr_{\pi}[C_{\pi} \subseteq S] - \Pr_{\pi}[C_{\pi} \subseteq S\setminus T].
    \]
    Thus, the only way for a circuit $C_{\pi}$ to contribute to $q_S$ and not $q_{S \setminus T}$ is if $C_{\pi}$ contains at least one element from $T$. Formally,
    \[
     q_S - q_{S \setminus T} = \Pr_{\pi}[C_{\pi} \subseteq S \wedge C_{\pi} \cap T \neq \emptyset].
    \]
    Thus, we obtain that 
    \[
    \frac{|T|}{2^{21} |S| \log(n)} \leq \Pr_{\pi}[C_{\pi} \subseteq S \wedge C_{\pi} \cap T \neq \emptyset],
    \]
    Finally, we can observe that 
    \[
    \Pr_{\pi}[C_{\pi} \subseteq S \wedge C_{\pi} \cap T \neq \emptyset] \leq p_{T, \M|_S}.
    \]
     This is because whenever we sample in accordance to a permutation $\pi$ over $E$ (the ground set of $\M$), and recover a circuit $C_{\pi}$ such that $T\cap C_{\pi}\neq \emptyset$ and $C_{\pi} \subseteq S$, the same permutation, if restricted to $S$ and used to sample elements of $S$, would have still given a circuit such that $T\cap S\neq \emptyset$.
     Together then, this means that 
     \[
     p_{T, \M|_S} \geq \frac{|T|}{2^{21} |S| \log(n)},
     \]
     as we desire.
\end{proof}

Next, we have the following claim:

\begin{claim}\label{clm:circuitMass}
    Suppose we invoke \cref{alg:globalOptimal} on a matroid $\M$, yielding a set $S$. Then,
    \[
    q_{S,\M} \geq 1-2^{-20}.
    \]
\end{claim}

\begin{proof}
    Consider an iteration of \cref{alg:globalOptimal} starting with a set $S$. We then recover a set $T$ such that 
    \[
    \hat{q}_{S} - \hat{q}_{S \setminus T} \leq \frac{|T|}{2^{20}|S|\log(n)},
    \]
    and set $S \leftarrow S \setminus T$. Inductively, we claim that when $S$ has $\ell$ elements remaining, that 
    \[
    \hat{q}_S \geq 1 - 2^{-20} + \sum_{i = 1}^{\ell} \frac{1}{2^{20} i \log(n)}.
    \]
    As a base case, we can consider when $S = E$, i.e., $|S| = n$. Then, we have that 
    \[
     1 - 2^{-20} + \sum_{i = 1}^{\ell} \frac{1}{2^{20} i \log(n)} \leq 1 - 2^{-20} + \frac{1 + \ln(n)}{2^{20} \log(n)} \leq 1 = \hat{q}_{S}.
    \]
    
    Now, consider an iteration where we have an intermediate set $S'$ of size $\ell'$, and subsequently remove a set $T$, yielding a set $S$ with $\ell$ elements. By induction, we suppose that 
    \[
    \hat{q}_{S'}\geq 1- 2^{-20} + \sum_{i = 1}^{\ell'} \frac{1}{2^{20} i \log(n)}.
    \]
    Now, for the set $T$ of elements that we remove (which is being removed from $S'$), we know that 
    \[
    \hat{q}_{S'} - \hat{q}_{S} \leq \frac{|T|}{2^{20}|S'|\log(n)}.
    \]
    This means that 
    \[
    \hat{q}_S \geq \hat{q}_{S'} - \frac{|T|}{2^{20} |S'| \log(n)} \geq 1 - 2^{-20} + \left ( \sum_{i = 1}^{\ell'} \frac{1}{2^{20} i \log(n)} \right ) - \frac{|T|}{2^{20} |S'| \log(n)}.
    \]
    Importantly, because $S = S' \setminus T$, we know that $|T| = \ell' - \ell$, so we can now re-write 
    \[
    \frac{|T|}{2^{20} |S'| \log(n)} = \sum_{i = \ell + 1}^{\ell'} \frac{1}{2^{20} \ell' \log(n)}.
    \]
    Plugging this in, we see that 
    \begin{align*}
    \hat{q}_S & \geq 1 - 2^{-20} + \left ( \sum_{i = 1}^{\ell} \frac{1}{2^{20} i \log(n)} \right ) + \left ( \sum_{i = \ell+1}^{\ell'}\frac{1}{2^{20} i \log(n)} - \frac{1}{2^{20} \ell' \log(n)} \right )\\
    & \geq 1 - 2^{-20} + \left ( \sum_{i = 1}^{\ell} \frac{1}{2^{20} i \log(n)} \right ) + \left ( \sum_{i = \ell+1}^{\ell'}\frac{1}{2^{20} \ell' \log(n)} - \frac{1}{2^{20} \ell' \log(n)} \right )\\
     & \geq 1-2^{20} + \left ( \sum_{i = 1}^{\ell} \frac{1}{2^{20} i \log(n)} \right ),
    \end{align*}
    thus yielding our inductive claim. 
    Finally, we use \cref{clm:boundedErrorEstimation} to conclude that $q_S \geq 1-2^{-20}$, given that $\hat{q}_S \geq 1-2^{-20} + \frac{1}{2^{20} \log(n)}$ (and assuming $n$ is sufficiently large). This yields the claim. 
\end{proof}

To summarize the above discussion, we now have the following lemma:

\begin{lemma}\label{lem:findGlobalOpt}
    \cref{alg:globalOptimal} returns a set $S$ which is globally-optimal with probability $1 - 2^{-\Omega(n)}$.
\end{lemma}

\begin{proof}
    \cref{clm:circuitMass} shows that $q_S \geq 1-2^{-20}$, and \cref{clm:hittingProbability} shows that $p_{T, \M|_S} \geq \frac{|T|}{2^{21} |S| \log(n)}$. The probability bound follows from \cref{clm:boundedErrorEstimation}.
\end{proof}

In the following subsection, we show how we can repeatedly invoke the above lemma by doing so-called \emph{iterative peeling}. In each round of peeling, we recover a new globally optimal set, thus creating the sequence of globally optimal sets guaranteed in \cref{lem:alphaGrowthFirstStatement}.

\subsection{Iterative Peeling}
\label{subsec:iterative_peeling}

As in \cite{KPS25}, our next step is to \emph{repeat} this procedure of peeling off sets. We present this algorithm below:

\begin{algorithm}[H]
\caption{RepeatedGlobalPeeling$(\M=(E,\I))$}\label{alg:iterativePeel}
    $\M\gets \text{RemoveSmallCircuits}(\M)$ \tcp{Algorithm 7 in \cite{KPS25}}
    $n \gets |E|$ \\
    $k \gets 0$ \\
    \While{$\M\neq \emptyset$}{
    $k\gets k +1$\\
    $S_k \leftarrow \mathrm{GloballyOptimalConstructor}(\M)$. \\
    $\M \leftarrow \M \setminus S_{k}$. \\
    \If{$\alpha(S_k)/|S_k|\geq 1/\log (n)$ or $|S_k|>n/2$} {
        \Return{$S_1,\dots S_{k-1}$}
    }
    }
    \Return{$S_1, \dots S_{k}$}
\end{algorithm}

\begin{remark}
At the beginning of \cref{alg:iterativePeel}, we invoke the RemoveSmallCircuits procedure (Algorithm 7 in \cite{KPS25}) to eliminate circuits of size $\leq 50$. This is for a minor technical reason in the probabilistic argument, the starting circuit size to be a sufficiently large constant. We omit the details here as they are identical to \cite{KPS25}.
\end{remark}

Notationally, we also let $\M_j = \M - S_1 - \dots - S_{j-1}$. While our globally-optimal sets are different than the notion of greedily-optimal sets used in \cite{KPS25}, our decomposition still enjoys all of the same favorable properties. Importantly, we have the following lemma which governs the growth of the $\alpha$-value within the sets that we peel off:

\begin{theorem} [Theorem 4.6 in \cite{KPS25}] \label{lem:alphaGrowth}
    Let $\mathcal{\M}$ be a matroid, and let $S_1, \dots S_k$ be a sequence of greedily-optimal sets that are peeled off (analogously to \cref{alg:iterativePeel}). Now, let $\ell\in [\log n]$ be an integer, let $T = \{ i \in [k]: |S_i| \in [2^{\ell}, 2^{\ell+1} -1]\}$, let $\gamma = |T|$, and let $a_1, \dots a_{\gamma}$ denote the indices in $T$. Then, with probability $1 - 2^{- \Omega(n)}$ it must be the case that 
\[
    \alpha(S_{a_i}) = \Omega(i^2) \text{ for every }i\in[\gamma], \quad \gamma = O\left(\sqrt{2^{\ell}}\right), \quad k=O(n^{1/3}).
\]
Additionally, for any $i < j \in [k]$,
\[
\frac{\alpha(S_j)}{|S_j|} = \Omega \left(\frac{\alpha(S_i)}{|S_i|}\right).
\]
\end{theorem}
Note that the proof of the final part of the claim above was actually \emph{not} presented in \cite{KPS25}. For this reason, we provide a complete proof of the above in \cref{sec:lem-proof}.

Finally, with this machinery in place, we can now prove \cref{lem:alphaGrowthFirstStatement}:

\begin{proof}
    First, we must show that each set $S_i$ that is recovered is globally-optimal with respect to $\M - S_1 - \dots S_{i-1}$. This follows by invoking \cref{lem:findGlobalOpt}, as the algorithm \cref{alg:iterativePeel} calls GloballyOptimalConstructor on $\M - S_1 - \dots S_{i-1}$ to construct $S_i$. 

    The remaining properties hold by first using \cref{rmk:globalisgreedy} to observe that our \emph{globally} optimal sets are also \emph{greedily} optimal sets, and then invoking \cref{lem:alphaGrowth} on this sequence of greedily optimal sets that is removed.
\end{proof}

\subsection{Finding Independent Sets via Globally-Optimal Sets}

With our notion of globally-optimal sets established, and an efficient algorithm for peeling them off, we also briefly mention here that the $\alpha$ value of each $S_j$ has a strong relationship with the $\alpha$ value of the parent matroid $\M_j$. Specifically:

\begin{claim}\label{clm:independentRecovery}
    Let $\M=(E,\I)$ be a matroid on $n$ elements, and let $S\subseteq E$ be a globally-optimal set in $M$. Then, for $\ell = \frac{\alpha(S)}{10|S|}n$, we have
\[
\Pr_\pi[\mathrm{Ind}(\{\pi(1),\dots,\pi(\ell)\})=1] \geq \frac{1}{4}.
\]
\end{claim}

\begin{proof}
    This follows by noting that a globally-optimal set is \emph{also} a greedily-optimal set as per \cref{rmk:globalisgreedy}. Then, we can simply invoke Claim 5.1 of \cite{KPS25}.
\end{proof}

Importantly, this gives the following lemma:

\begin{lemma}\label{lem:independentProg}
    There is a $1$-round, polynomial-query algorithm which, given a matroid $\M$ on $n$ elements and globally-optimal set $S$, recovers $\Omega\left(\frac{\alpha(S)}{|S|}n\right)$ independent elements with probability $1 - 2^{-n}$.
\end{lemma}

\begin{proof}
    Simply sample polynomially many random permutations in accordance with \cref{clm:independentRecovery}. With exponentially high probability, one of these random permutations will find $\ell = \frac{\alpha(S)}{10|S|}n$ independent elements.
\end{proof}

Now, in the next section, we show how to use our decomposition to obtain improved algorithms for recovering \emph{redundant} elements as well.

\section{Better Progress with Circuit Size and Marginal Probability Trade-Offs}\label{sec:circuitProbs}

In the previous section, we established an algorithm for recovering globally-optimal sets, and showed some basic properties about how the $\alpha$-value of globally-optimal sets lends itself towards recovering large independent sets in the parent matroid. In this section, we instead focus on our ability to \emph{recover redundant elements} within each globally-optimal set we recover in \cref{alg:iterativePeel}. Ultimately, we show the following:

\begin{theorem}\label{thm:balancedProgress}
    Let $S$ be a globally optimal set recovered from the matroid $\M$. Then there is a single round procedure which deletes $\widetilde{\Omega}\left (\min\left(|S|, \frac{|S|^{3/2}}{\alpha(S)}\right)\right)$ redundant elements with probability $1 - 2^{- \Omega(n)}$.
\end{theorem}

Before proving this theorem, we first observe an implicit strengthening of \cite{KPS25}, Lemma 6.7 which we will use to efficiently recover redundant elements:

\begin{claim}[\cite{KPS25}]\label{clm:basicRecoverRedundant}
Let $\M=(E,\I)$ be a matroid and $S\subseteq E$, such that for every $x \in S$, $p_{x, \M|_S} \geq 1 / n^2$. Then, there is a one round algorithm making polynomially many queries which recovers 
\[
\widetilde{\Omega} \left ( \sum_{x \in S} \min \left (1, p_{x, \M|_S} \cdot \frac{|S|^2}{\alpha(S)^2} \right ) \right )
\]
redundant elements in $S$ with probability $1 - 2^{- \Omega(n)}$.
\end{claim}

Note that this statement is actually slightly stronger than the statement that appears in \cite{KPS25}, Lemma 6.7, and only appears implicitly in their proof. For this reason, we include a formal re-derivation in \cref{sec:appendixS2alpha2}.

\subsection{Performing the Core Decomposition}

Next, as mentioned in \cref{sec:introcore}, we require a notion which captures the variation in marginal element probabilities in a matroid. To do so, we introduce the notion of core elements in the matroid:

\begin{definition}
    Let $\M=(E,\I)$ be a matroid and $S\subseteq E$. We let $\mathrm{CORE}_S = \{x \in S: p_{x, \M|_S} \geq \frac{\alpha(S)^2}{|S|^2} \}.$
\end{definition}

Any core element has a large marginal probability, and so \cref{clm:basicRecoverRedundant} guarantees that we can find many redundant elements:

\begin{claim}\label{clm:coreRecovery}
   Let $\M=(E,\I)$ be a matroid and $S\subseteq E$ a globally-optimal set. Then, there is a one round sub-routine making polynomially many queries which recovers 
    \[
    \widetilde{\Omega}(|\mathrm{CORE}_S|)
    \]
    redundant elements with probability $1 - 2^{-\Omega(n)}$.
\end{claim}

\begin{proof}
    Observe that for every element $x \in \mathrm{CORE}_S$, we have 
    \[
    p_{x, \M|_S} \geq \frac{\alpha(S)^2}{|S|^2}.
    \]
    By \cref{clm:basicRecoverRedundant} (and using the fact that $S$ is globally-optimal, thereby satisfying the marginal probability requirement) we know that we can recover 
    \[
    \widetilde{\Omega} \left ( \sum_{i \in S} \min \left (1, p_{i, \M|_S} \cdot \frac{|S|^2}{\alpha(S)^2} \right ) \right ) = \widetilde{\Omega} \left ( \sum_{i \in \mathrm{CORE}_S} \min \left (1, p_{i, \M|_S} \cdot \frac{|S|^2}{\alpha(S)^2} \right ) \right )
    \]
    \[
    = \widetilde{\Omega} \left ( \sum_{i \in \mathrm{CORE}_S} 1 \right ) = \widetilde{\Omega} \left ( |\mathrm{CORE}_S| \right )
    \]
    redundant elements in a single round using only polynomially many queries, thus yielding our claim. 
\end{proof}

In particular, in the remainder of this section, we will assume that $|\mathrm{CORE}_S| \leq |S| / 2$, as otherwise we are guaranteed to be able to find $\widetilde{\Omega} \left ( |\mathrm{CORE}_S| \right )$ redundant elements from $S$ in a single round. Importantly, this means that there are at least $|S| / 2$ \emph{non-core} elements. We denote this set of elements by $S \setminus \mathrm{CORE}_S$.

Immediately, we have the following claim:

\begin{claim}\label{clm:nonCoreMassRecovery}
    Let $\M=(E,\I)$ be a matroid and $S\subseteq E$. Then, there is a one round sub-routine making polynomially many queries which recovers 
    \[
    \widetilde{\Omega} \left ( \sum_{i \in S \setminus \mathrm{CORE}_S} p_{i, \M|_S} \cdot \frac{|S|^2}{\alpha(S)^2}\right )
    \]
    redundant elements with probability $1 - 2^{- \Omega(n)}$.
\end{claim}

\begin{proof}
    As before, we use \cref{clm:basicRecoverRedundant}, with the only modification being that now for $i \in S \setminus \mathrm{CORE}_S$, 
    \[
    \min \left (1, p_{i, \M|_S} \cdot \frac{|S|^2}{\alpha(S)^2} \right) = p_{i, \M|_S} \cdot \frac{|S|^2}{\alpha(S)^2}.
    \]
    Thus, we can recover 
    \[
     \widetilde{\Omega} \left ( \sum_{i \in S} \min \left (1, p_{i, \M|_S} \cdot \frac{|S|^2}{\alpha(S)^2} \right ) \right ) = \widetilde{\Omega} \left ( \sum_{i \in S \setminus \mathrm{CORE}_S}  p_{i, \M|_S} \cdot \frac{|S|^2}{\alpha(S)^2}  \right )
    \]
    redundant elements in $S$, as we desire. 
\end{proof}

As mentioned in \cref{sec:introcore}, ultimately we will create a \emph{trade-off} between the probability mass contained in non-core elements, and the circuit sizes that these non-core elements participate in. To show why such a dichotomy is useful, we have the following generic statement which showcases our ability to make progress by recovering short circuits:

\begin{claim}\label{clm:deleteElementsCircuit}
    Let $W \subseteq S \setminus \mathrm{CORE}_S$ denote a set of elements such that for every $x \in W$, we have recovered a circuit $\mathrm{Circ}_x$ for which $|\mathrm{Circ}_x \cap (S \setminus \mathrm{CORE}_S)| \leq \ell$. Then, we can find a set of $\Omega\left ( \frac{|W|}{\ell}\right )$ redundant elements to delete. 
\end{claim}

\begin{proof}
    Observe that for every element $x$, $\mathrm{Circ}_x - \{x\}$ witnesses $x$ being a redundant element. To identify our set of elements to delete, we first commit to keeping \emph{all} of the elements in $\bar{W} = S \setminus W$. Thus, the only deletions will be coming from the set $W$. So, the only dependencies we must be concerned with respecting are the dependencies of $\mathrm{Circ}_x$ with $W$ (i.e., the elements in $\mathrm{Circ}_x \cap \bar{W}$). So, we construct a bipartite graph where $L = R = W$. For every left-hand vertex $x$, the corresponding right neighborhood is exactly $\mathrm{Circ}_x \cap \bar{W} - \{x\}$. In this sense, the neighborhood of $x$ captures a sufficient set of elements to witness the redundancy of element $x$.

   Now, we can see that $|\Gamma(x)| = |\mathrm{Circ}_x \cap \bar{W} - \{x \}| \leq |\mathrm{Circ}_x \cap (S - \mathrm{CORE}_S)| \leq \ell$. Thus, we consider a simple peeling procedure for finding a maximal set of dependent elements: we start by setting the set $D = \emptyset$ and $K = \emptyset$ ($D$ will be the elements we delete, namely the redundant elements, and $K$ is the elements we keep). 
   
   In the first round, we start with the first vertex remaining in $L$: denote this vertex by $v$. We then delete $v$ from $L$ and $R$, and add $v$ to $D$. We then delete all vertices in $\Gamma(v)$ from $L$ and $R$, but add $\Gamma(v)$ to $K$. Importantly then, we have the property that $v \in \mathrm{Span}(K \cup \bar{W})$, we $v$ formed a circuit with $\Gamma(v) \cup W$, and is thus in its span. 

   Now, we continue this procedure. In each round, we remove an arbitrary vertex $v$ from $L$ and add it to $D$, and then remove all vertices in $\Gamma(v)$ and add them to $K$. Inductively, we claim that $\forall v \in D$, $v \in \mathrm{Span}(K \cup \bar{W})$. Let us suppose this is true for the first $i$ rounds of the above procedure, and let us analyze the $i+1$st round. We let $D_i$ denote the elements deleted in the first $i$ rounds. By induction, we can see that $D_i \in \mathrm{Span}(K \cup \bar{W})$. Thus,
   \[
   \mathrm{Span}(K \cup \bar{W}) = \mathrm{Span}(K \cup \bar{W} \cup D_i).
   \]
   Now, let $v$ denote the vertex we remove in the $i+1$st round, and let $K$ denote the resulting set of elements that we keep (i.e., adding $\Gamma(v)$ to $K_i$). We claim that $\forall y \in \mathrm{Circ}_v - \{v\}$, $y \in \mathrm{Span}(K \cup \bar{W})$, and thus $v \in \mathrm{Span}(K \cup \bar{W})$. To see this, consider any  element in $\mathrm{Circ}_v - \{v\}$
   Now, let $v$ denote the vertex we remove in the $i+1$st round, and let $K$ denote the resulting set of elements that we keep (i.e., adding $\Gamma(v)$ to $K_i$). We claim that $\forall y \in \mathrm{Circ}_v - \{v\}$, $y \in \mathrm{Span}(K \cup \bar{W})$, and thus $v \in \mathrm{Span}(K \cup \bar{W})$. To see this, consider any  element $y$ in $\mathrm{Circ}_v - \{v\}$. There are several cases:
   \begin{enumerate}
       \item If $y \in \Gamma(v)$, then $y \in K$, and so trivially $y \in \mathrm{Span}(K \cup \bar{W})$.
       \item If $y \in \bar{W}$, then again trivially $y \in \mathrm{Span}(K \cup \bar{W})$.
       \item If $y \notin \Gamma(v)$ and $y \in W$, then $y$ must have been deleted in an earlier round. Thus $y$ was either already added to $K$, in which case $y \in \mathrm{Span}(K \cup \bar{W})$, or $y \in D_i$, in which case $y \in \mathrm{Span}(K \cup \bar{W})$ by our inductive hypothesis. 
   \end{enumerate}
   Thus, $\forall y \in \mathrm{Circ}_v - \{v\}$, $y \in \mathrm{Span}(K \cup \bar{W})$, and so $v \in \mathrm{Span}(K \cup \bar{W})$ as well. 

   Thus, every element in the set $D$ is redundant, and it remains only to bound the size of the set $D$. For this, observe that every time we delete $\Gamma(v)$, we delete at most $\ell$ elements from $L$. Thus, we can repeat this procedure for at least $|W| / \ell$ iterations, adding one element to $D$ in each round. Thus, $|D| = \Omega \left ( \frac{|W|}{\ell}\right )$, as we desire.
\end{proof}

\subsection{Finding Short Circuits for Non-Core Elements}

The final claim in the previous subsection shows that \emph{if we can find} short circuits, then there is a simple algorithm for recovering many redundant elements. In this section, we show how we can indeed find them.

As a first step towards finding these short circuits, below we show that we can efficiently estimate some statistics about these circuits by performing random sampling:

\begin{claim}\label{clm:estimateCircuitStats}
    Let $\M=(E,\I)$ be a matroid and $S\subseteq E$ such that for every $i \in S$, $p_{i, \M|_S} = \widetilde{\Omega}(1 / |S|)$. Then, there is a one round algorithm which for every $i \in S, T \subseteq S$ estimates
    \[
    w_{i, S, T} = \E_{\pi^{(1)}, \dots \pi^{(n^{10})}: i \in C_{\pi^{(1)}}, \dots i \in C_{\pi^{(n^{10})}}}[\min_{j \in [n^{10}]}|C_{\pi^{(j)}} \cap (S\setminus T)|]
    \]
    to additive error $1$, with probability $1 - 2^{- \Omega(n)}$. Additionally, this algorithm finds a circuit $\mathrm{Circ}_{i, T}$ which satisfies 
    \[
    |\mathrm{Circ}_{i, T} \cap (S\setminus T)| \leq w_{i, S, T} + 1.
    \]
\end{claim}

\begin{proof}
    The algorithm simply samples many random groups of $n^{10}$ permutations $\pi$ and heuristically estimates $w_{i, S, T}$, for every $i, T$. For each $i$, a random permutation $\pi$ satisfies $i \in C_{\pi}$ with probability $p_{i,\M|_S} = \widetilde{\Omega}(1 / |S|)$, and thus sampling $\mathrm{poly}(|S|) \leq \mathrm{poly}(n)$ many permutations $\pi$ yields $\mathrm{poly}(n)$ uniformly random permutations in the support of each expectation. Now, since 
    \[
    w_{i, S, T} = \E_{\pi^{(1)}, \dots \pi^{(n^{10})}: i \in C_{\pi^{(1)}}, \dots i \in C_{\pi^{(n^{10})}}}[\min_{j \in [n^{10}]}|C_{\pi^{(j)}} \cap (S\setminus T)|] \in [0,n],
    \]
    a simple Chernoff bound yields that we can empirically estimate $w_{i, S, T}$ to additive error $1$ with probability $1 - 2^{-n}$ by sampling $\leq n^{10}$ random such groups of permutations.

    To see the constructive aspect of the result, observe that because our empirical estimates $\hat{w}_{i, S, T} \leq w_{i, S, T} + 1$, there must be one explicit circuit we recover which satisfies $|\mathrm{Circ}_{i, T} \cap (S\setminus T)| \leq w_{i, S, T} + 1$.
\end{proof}

Finally, we formally prove the tradeoff between the probability mass of the non-core elements and the circuit sizes we recover: 

\begin{lemma}\label{lem:circuitRedundantTradeoff}
    Let $\M=(E,\I)$ be a matroid and $S\subseteq E$ be globally optimal. Let $\ell = \sum_{i \in S - \mathrm{CORE}_S} p_{i, \M|_S}$. Then, there is a one round sub-routine using polynomially many queries which, with probability $1 - 2^{-\Omega(n)}$, recovers:
    \begin{enumerate}
        \item $\widetilde{\Omega}\left (\min\left(|S|, \frac{\ell \cdot |S|^2}{\alpha(S)^2}\right) \right )$ redundant elements. 
        \item $\widetilde{\Omega} \left ( \frac{|S|}{\ell} \right )$ redundant elements.
    \end{enumerate}
\end{lemma}

\begin{proof}
    To see the first item, we can simply plug in $\ell = \sum_{i \in S - \mathrm{CORE}_S} p_{i,\M|_S}$ to \cref{clm:nonCoreMassRecovery}. 

    To see the second item, recall that if $|\mathrm{CORE}_S| \geq |S| / 2$, then by \cref{clm:coreRecovery}, we can recover $\widetilde{\Omega}(|S|)$ redundant elements. So, we assume that $|S - \mathrm{CORE}_S| \geq |S| / 2$.

    Next, we observe that
    \[
    \ell = \sum_{i \in S - \mathrm{CORE}_S} p_{i, \M|_S} = \sum_{i \in S - \mathrm{CORE}_S} \Pr_{\pi}[i \in C_{\pi}] = \E_{\pi}[|(S - \mathrm{CORE}_S)| \cap C_{\pi}].
    \]
    In words, $\ell$ is exactly the number of elements we expect to see in each circuit that come from the non-core elements. Now, let us use the one round sub-routine from \cref{clm:estimateCircuitStats}, and let $Q$ denote the set of elements $i$ for which $|\mathrm{Circ}_{i, \mathrm{CORE}_S} \cap (S - \mathrm{CORE}_S)| \geq 2^{26} \ell \log(n)$. This implies that for all $i \in Q$, $w_{i, S, \mathrm{CORE}_S} \geq 2^{26} \ell \log(n) - 1 \geq 2^{25} \ell \log(n)$. We claim that $|Q| \leq |S| / 4$. Indeed, suppose for the sake of contradiction that $|Q| > |S|/4$. Then, because $S$ is globally optimal, it is the case that $p_{Q, \M|_S} \geq \frac{|Q|}{2^{21}|S|\log(n)} \geq \frac{1}{2^{23}\log(n)}$. Now, we claim that 
    \begin{claim}\label{clm:expectedQIntersect}
    \[
    \E_{\pi: C_{\pi} \cap Q \neq \emptyset}[|C_{\pi} \cap (S - \mathrm{CORE}_S)|] > 2^{24} \ell \log(n).
    \]
    \end{claim}
    \begin{proof}[Proof of \cref{clm:expectedQIntersect}]
    Indeed, if we suppose for the sake of contradiction that 
        \[
    \E_{\pi: C_{\pi} \cap Q \neq \emptyset}[|C_{\pi} \cap (S - \mathrm{CORE}_S)|] \leq 2^{24} \ell \log(n),
    \]
    then by Markov's inequality, with probability $\geq \frac{1}{2}$, when sampling a permutation $\pi$ such that $C_{\pi} \cap Q \neq \emptyset$, we have that $|C_{\pi} \cap (S - \mathrm{CORE}_S)| \leq 2^{25} \ell \log(n)$. Because $Q \subseteq S \subseteq E$, we know that $|Q| \leq n$. By a pigeonhole argument, we know that there exists an element $i \in Q$ then such that 
    \[
    \Pr_{\pi: C_{\pi} \cap Q \neq \emptyset}[i \in C_{\pi} \wedge |C_{\pi} \cap (S - \mathrm{CORE}_S)| \leq 2^{25} \ell \log(n)] \geq \frac{1}{2n}.
    \]
    Thus, we get that 
    \[
    \Pr_{\pi: C_{\pi} \cap Q \neq \emptyset}[i \in C_{\pi}] \cdot  \Pr_{\pi: C_{\pi} \cap Q \neq \emptyset} [|C_{\pi} \cap (S - \mathrm{CORE}_S)| \leq 2^{25} \ell \log(n) \mid i \in C_\pi] \geq \frac{1}{2n},
    \]
    which means that 
    \[
    \Pr_{\pi: C_{\pi} \cap Q \neq \emptyset} [|C_{\pi} \cap (S - \mathrm{CORE}_S)| \leq 2^{25} \ell \log(n) \mid i \in C_\pi] \]
    \[
    = \Pr_{\pi: i \in C_{\pi} } [|C_{\pi} \cap (S - \mathrm{CORE}_S)| \leq 2^{25} \ell \log(n)]\geq \frac{1}{2n}.
    \]
    However, the above implies that
    \[
    w_{i, S, \mathrm{CORE}_S} = \E_{\pi^{(1)}, \dots \pi^{(n^{10})}: i \in C_{\pi^{(1)}}, \dots i \in C_{\pi^{(n^{10})}}}[\min_{j \in [n^{10}]}|C_{\pi^{(j)}} \cap (S - \mathrm{CORE}_S)|] \leq  2^{25} \ell \log(n) + 1,
    \]
    as there is a $(1 - \frac{1}{2n})^{n^{10}} \leq 2^{-n}$ probability of not seeing a circuit of length $\leq 2^{25} \ell \log(n)$ when sampling $n^{10}$ circuits that contain element $i$. Thus, $w_{i, S, \mathrm{CORE}_S} \leq 2^{25} \ell \log(n) (1 - 2^{-n}) + n \cdot 2^{-n} \leq 2^{25} \ell \log(n) + 1$.
    However, this contradicts our choice of $Q$, as we defined it such that for all $i \in Q$, $w_{i, S, \mathrm{CORE}_S} \geq 2^{26} \ell \log(n)$.
    \end{proof}

    So, \cref{clm:expectedQIntersect} means
    \[
    \E_{\pi: C_{\pi} \cap Q \neq \emptyset}[|C_{\pi} \cap (S - \mathrm{CORE}_S)|] > 2^{24} \ell \log(n).
    \]
    We then get the simple bound that 
    \[
    \E_{\pi}[|(S - \mathrm{CORE}_S) \cap C_{\pi}|] \geq \Pr_{\pi}[Q \cap C_{\pi} \neq \emptyset] \cdot \E_{\pi}[\left |(S - \mathrm{CORE}_S) \cap C_{\pi} \right |\mid Q \cap C_{\pi} \neq \emptyset] 
    \]
    \[
    = p_{Q, \M|_S} \cdot  \E_{\pi:Q \cap C_{\pi} \neq \emptyset}[|(S - \mathrm{CORE}_S) \cap C_{\pi}|] > p_{Q, \M|_S} \cdot (2^{24} \ell \log(n)).
    \]
    Recall that if $|Q| > |S| / 4$, this means that $p_{Q, \M|_S} \geq \frac{1}{2^{23} \log(n)}$, so
    \[
    \E_{\pi}[|(S - \mathrm{CORE}_S)| \cap C_{\pi}] \geq \frac{1}{2^{23} \log(n)} \cdot (2^{24} \ell \log(n)) > \ell
    \]
    which yields a contradiction with the fact that 
    \[
    \ell = \sum_{i \in S - \mathrm{CORE}_S} p_{i, \M|_S} = \sum_{i \in S - \mathrm{CORE}_S} \Pr_{\pi}[i \in C_{\pi}] = \E_{\pi}[|(S - \mathrm{CORE}_S)\cap C_{\pi}| ].
    \]
    Hence $|Q| \leq |S|/4$.

    Finally, this means that for $Z = S - \mathrm{CORE}_S - Q$, $|Z| = \Omega( |S|)$, and for every element $i \in Z$, we can recover via \cref{clm:estimateCircuitStats} a circuit $\mathrm{Circ}_{i, \mathrm{CORE}_S}$ such that $|\mathrm{Circ}_{i, \mathrm{CORE}_S} \cap (S - Z)| \leq |\mathrm{Circ}_{i, \mathrm{CORE}_S} \cap (S - \mathrm{CORE}_S)| < 2^{26} \ell \log(n)$. By \cref{clm:deleteElementsCircuit}, this means that we can recover $\widetilde{\Omega}(|S| / \ell)$ redundant elements in a single round. The probability bound follows because each individual claim holds with probability $1 - 2^{- \Omega(n)}$.
\end{proof}

\subsection{Proof of \cref{thm:balancedProgress}}

Finally, we can present the proof of our main theorem in this section, by simply conditioning on the value of $\ell$ in \cref{lem:circuitRedundantTradeoff}:

\begin{proof}[Proof of \cref{thm:balancedProgress}]
     \cref{lem:circuitRedundantTradeoff} guarantees that we can always delete 
     \[
     \widetilde{\Omega} \left (\max \left ( \min\left(|S|, \frac{\ell \cdot |S|^2}{\alpha(S)^2}\right), \frac{|S|}{\ell}\right ) \right )
     \]
     redundant elements, where $\ell$ is the probability mass in the non-core elements. In particular, the worst case behavior is when 
     \[
    \frac{\ell \cdot |S|^2}{\alpha(S)^2} = \frac{|S|}{\ell},
     \]
     which occurs when $\ell = \frac{\alpha(S)}{\sqrt{|S|}}$. For this value, we are only guaranteed to delete 
     \[
     \widetilde{\Omega}\left (\min\left(|S|, \frac{|S|^{3/2}}{\alpha(S)}\right)\right)
     \]
     redundant elements, as claimed. The probability bound follows from the probability bound of \cref{lem:circuitRedundantTradeoff}.
\end{proof}

\section{An Algorithm with $\widetilde{O}(n^{3/7})$ Round Complexity}\label{sec:analysisProg}

In the previous sections, we developed a new matroid decomposition that iteratively peels off globally optimal sets, and established conditions under which for any such set $S$, we can make progress by either identifying a large independent subset to contract or a large number of redundant elements to delete. In this section, we synthesize these components into a single algorithm that carefully utilizes these two forms of progress, and finds a basis in $\widetilde{O}(n^{3/7})$ rounds.
We will prove here our main theorem:

\begin{theorem}\label{thm:3-7rounds}
    For an arbitrary matroid $\M$ on $n$ elements, there is an $\widetilde{O}(n^{3/7})$ round algorithm, making polynomially many independence queries, which recovers a basis of $\M$ with high probability. 
\end{theorem}

\subsection{A Guaranteed Progress Algorithm}

Our strategy does not guarantee substantial progress in every single round. Instead, we adopt an amortized analysis framework and show that while some rounds may be spent waiting for right structural conditions to develop, the algorithm must eventually terminate by achieving a target {\em average progress} per round.
Specifically, we define an {\em average progress} target
of $f = n^{4/7} / \log^c(n)$ for a sufficiently large constant $c$. Progress is measured as the number of elements we can either contract (as an independent set) or delete (as redundant elements). Our main algorithm, presented below, runs the peeling process and tracks cumulative progress. It terminates as soon as the average progress per round meets or exceeds $f$.
Clearly, this rate of progress is sufficient to find a basis in $\widetilde{O}(n^{3/7})$ rounds.

\begin{algorithm}[H]
\caption{GuaranteedProgressDecomposition$(\M=(E,\I))$}\label{alg:new-decomp-2}
    $\M\gets \text{RemoveSmallCircuits}(\M)$.\\
    $T\gets 0$.\\
    $i \gets 0$.\\
    \While{$|E|\geq n/2$} {
        $i\gets i+1$.\\
        $S_i\gets \text{GloballyOptimalConstructor}(\M)$. \\
        Let $\widehat \alpha(S_i)$ be the estimation of $\alpha(S_i)$, which satisfies $(\alpha(S_i)-1)/2\leq \widehat \alpha(S_i)\leq 2\alpha(S_i)$ with high probability. \tcp{see Definition 4.6 and Claim 4.8 of \cite{KPS25}}
        \If{$\frac{\widehat \alpha(S_i)}{|S_i|} n \geq i\cdot f$} {
            \Return \label{line:return-contract}
        }
        \If{$S_i$ is good} {
            $T\gets T+|S_i|$.
        }
        \Else{$T\gets T+\frac{|S_i|^{3/2}}{\widehat\alpha(S_i)}$.}
        \If{$T \geq i\cdot f$} {
            \Return \label{line:return-delete}
        }
        $\M\gets \M\setminus S_i$.
    }
\end{algorithm}

Our algorithm relies on categorizing each peeled set as either {\em good} or {\em bad}, as defined below. 

\begin{definition}
    Let the sets returned by \cref{alg:new-decomp-2} be denoted by $S_1, \dots S_{m}$, we say that a set $S_{i}$ is \textbf{good} if $p_{x,\M|_{S_{i}}} \geq \alpha(S_{i})^2 / |S_{i}|$ for more than half of the elements in $x\in S_{i}$, as we can then recover $\widetilde{\Omega}(|S_{i}|)$ redundant elements in a single additional round.

    We say that $S_i$ is \textbf{bad} if the above does not hold, as we can then only recover $\widetilde{\Omega}\left ( \frac{|S_{i}|^{3/2}}{\alpha(S_{i})} \right )$ redundant elements, as per \cref{thm:balancedProgress}.
\end{definition}

Clearly, if the algorithm returns on \cref{line:return-contract} or \cref{line:return-delete}, we make $\widetilde \Omega(f)$ progress per round on average. So, in the following, we prove that the algorithm will indeed always return on \cref{line:return-contract} or \cref{line:return-delete}.

\subsection{Analysis of Guaranteed Progress Decomposition Algorithm}

We will analyze \Cref{alg:new-decomp-2} by carefully tracking the sets peeled off during the decomposition and categorizing them both by size and progress type, namely, good or bad. For bad sets, those where the immediate progress is suboptimal, we group them by size and focus on the most populous (dominant) category. By establishing lower bounds on the $\alpha$-values of these bad sets and bounding how often the dominant category can change, we show that sustained periods of suboptimal progress necessarily lead to a rapid increase in the $\alpha$-parameter. This, in turn, forces a transition to rounds with much higher progress, preventing the algorithm from stagnating.

The rest of the section formalizes this plan, introducing the necessary definitions, invariants, and potential-function style arguments to make the amortized analysis rigorous.

For the sake of contradiction, suppose that the algorithm does not terminate on \cref{line:return-contract} or \cref{line:return-delete}, and in the process peels off sets $S_1, \dots S_m$. We first observe that for every $i\in [m]$, it must be the case that $\alpha(S_i)/|S_i|<1/\log (n)$. Therefore, \cref{lem:alphaGrowth} applies for $S_1,\dots S_m$.

\begin{claim}
For any $i\in [m]$, $\alpha(S_i)/|S_i|<1/\log (n)$.
\end{claim}
\begin{proof}
For contradiction, let $S_i$ be the first set that $\alpha(S_i)/|S_i|\geq 1/\log (n)$. Then we can invoke \cref{lem:alphaGrowth} for $S_1,\dots S_{i-1}$ and obtain that $i=O(n^{1/3})$.

If $\alpha(S_i)/|S_i|\geq 1/\log (n)$, then we have
\[
\frac{\widehat \alpha(S_i)}{|S_i|} n = \Omega\left(\frac{\alpha(S_i)}{|S_i|}n\right) = \Omega\left(\frac{n}{\log n}\right) \geq i\cdot f
\]
as $i\cdot f = O(n^{1/3}\cdot n^{4/7})=O(n^{19/21})$. Therefore, the algorithm should have returned on \cref{line:return-contract}, which is a contradiction.
\end{proof}

Next, we introduce some additional terminology.

\begin{definition}
    For any $\ell \in [\log(n)]$ and any $i \in [m]$, we let
    \[
    T_{\ell, i} = \{j \leq i: |S_j| \in [2^{\ell-1},2^{\ell}], S_j \text{ is bad} \}.
    \]
    That is, $T_{\ell, i}$ is the prefix of bad sets before $i$ whose size is in the range $[2^{\ell-1},2^{\ell}]$.
\end{definition}

In our analysis, we will frequently index into these bad sets of a specific size.
We note that grouping by size is crucial because our progress lower bounds and the evolution of $\alpha$-values are tightly linked to $|S_j|$.

\begin{definition}
    For $\ell \in [\log(n)]$ and $j \in [|T_{\ell,m}|]$, we let $r_{\ell, j}$ be the index of the $j$th bad set of size $[2^{\ell-1},2^{\ell}]$. We let $R_{\ell, j} = S_{r_{\ell, j}}$ denote the actual bad set itself. 
\end{definition}

\paragraph{Dominant Categories and Consistent Chains.}
At each index $i \in [m]$, we will be primarily interested in the \emph{size} of bad sets that is most common in the prefix of sets before $S_i$. 

\begin{definition}\label{def:di}
    For any $i \in [m]$, we let 
    \[
    d_i = \mathrm{argmax}_{\ell \in [\log(n)]} |T_{\ell, i}|
    \]
    be the \textbf{dominant category}. Additionally, we add the restriction that that $d_{i +1} \neq d_i$ \emph{only} if the $i+1$st set is of size $\in [2^{d_{i+1} -1} + 1, 2^{d_{i+1}}]$.\footnote{That is to say, when there are ties between different possible choices of the dominant category, $d_{i+1}$ either stays the same, or assumes the category of the most recently added element. }
\end{definition}

This definition ensures that among all bad sets before $S_i$, at least a $(1 / \log(n))$-fraction of them belong to the dominant category $d_i$. 

We are interested in iterations $i$ when the dominant category changes, that is, when $d_i \neq d_{i+1}$. Before diving into this, we need to introduce notion of good and bad indices.

\begin{definition}
    For an index $i \in [m]$, we say that $i$ is a \textbf{good index} if at least half of the sets $S_1, \dots S_m$ are good. Otherwise, we say that $i$ is a \textbf{bad index}. We let $m' \in [m]$ denote the final good index. 
\end{definition}

\begin{definition}
    Let $K = \{i: d_{i} \neq d_{i+1}: m' \leq i \leq m \}$ be the \textbf{set of indices where the dominant category changes}. We let $k = |K|$, and let $a_i$ denote the $i$th element in $K$, with $a_0 = m'$.
\end{definition}

Ultimately, we will be concerned with the sequences of sets in between consecutive dominant category changes:

\begin{definition}
    Between consecutive dominant category changes $a_{i-1}$ and $a_{i}$, we refer to $S_{a_{i-1}+1}, \dots S_{a_{i}}$ as a \textbf{consistent chain with dominant category} $\ell_i = d_{a_{i-1}+1}=\dots=d_{a_i}$. We let $b_i = 2^{\ell_i}$ denote the size of the dominant category in this consistent chain. Finally, we let $p_i$ be the smallest index such that $r_{\ell_i, p_i} > a_{i-1}$, and $q_i$ the largest index such that $r_{\ell_i, q_i} \leq a_{i}$.
\end{definition}

A consistent chain is thus a maximal subsequence of bad iterations where the dominant category does not change. Our analysis crucially relies on analyzing progress over consistent chains.

We will rely on the following simple claim which relates the number of bad sets of the dominant category to the actual indices among the bad sets.

\begin{claim}\label{clm:number-of-rounds}
    For any $i \in [k]$, we have:
    \begin{enumerate}
        \item For any $j \in [p_i, q_i]$, $r_{\ell_i, j} = \widetilde{O}(j)$.
        \item $a_i = \widetilde{O}(q_i)$.
        \item $r_{\ell_i, p_i} = a_{i-1}+1$.
    \end{enumerate}
\end{claim}

\begin{proof}
    By definition, for every $t \in [a_{i-1} + 1, a_i]$, $d_t  = \ell_i$ is the dominant category. So, for every such $t$, we know $|T_{\ell_i, t}| \geq |T_{\ell', t}|$, for every $\ell' \neq \ell_i \in [\log(n)]$. In particular, this means that the total number of bad sets before index $t$ is at most 
    \[
    \sum_{\ell \in [\log(n)]} |T_{\ell, t}| \leq |T_{\ell_i, t}| \cdot \log(n).
    \]
    Because at least half of the sets before $t$ are bad, we know that 
    \[
    \frac{t}{2} \leq \sum_{\ell \in [\log(n)]} |T_{\ell, t}| \leq |T_{\ell_i, t}| \cdot \log(n),
    \]
    which means that $t = \widetilde{O}(|T_{\ell_i, t}|)$. 
    \begin{enumerate}
        \item If we set $t = r_{\ell_i, j}$ for $j \in [p_i, q_i]$, then $|T_{\ell_i, t}| = j$, so we obtain $r_{\ell_i, j} = \widetilde{O}(j)$.
        \item If we set $t = a_i$, then $|T_{\ell_i, t}| = q_i$, so we get $a_i = \widetilde{O}(q_i)$.
    \end{enumerate}

    To see the third item of the claim, observe that by \cref{def:di}, the smallest index $p_i$ such that $r_{\ell_i, p_i} > a_{i-1}$ is always $a_{i-1} + 1$, as whenever the dominant category changes, we always insist that the next set's category \emph{is} the dominant category. 
    
    This yields the claim. 
\end{proof}

\paragraph{Bounding $\alpha$-values in Consistent Chains.}
Now, using the above claim, we can derive a strong bound on the growth of the $\alpha$-value in the sets of dominant category that we peel off, using our bound on the average progress per round.

\begin{claim}\label{clm:alpha-lower}
    For any $i \in [k]$, and any $j \in [p_i, q_i]$, 
    \[
    \alpha(R_{\ell_i, j}) = \widetilde{\Omega}\left ( \frac{b_i^{3/2}}{f} \right ).
    \]
\end{claim}

\begin{proof}
    For every $j \in [p_i, q_i]$, we have that 
\[
\sum_{t=1}^j \frac{|R_{\ell_i,t}|^{3/2}}{\alpha(R_{\ell_i,t})} = \sum_{t=1}^j \frac{\Theta(b_i^{3/2})}{\alpha(R_{\ell_i,t})} = \Omega\left( j \cdot \frac{b_i^{3/2}}{\alpha(R_{\ell_i,j})}\right),
\]
where we have used that $|R_{\ell_i, t}| = \Theta(b_i)$ by definition, and that by \cref{lem:alphaGrowth}, $\alpha(R_{\ell_i, t}) = O( \alpha(R_{\ell_i,j}))$.

At the same time, because we assume the algorithm does not terminate early, it must be the case that 
\[
\sum_{t=1}^j \frac{|R_{\ell_i,t}|^{3/2}}{\alpha(R_{\ell_i,t})} \leq r_{\ell_i,j}\cdot f = \tilde O(j\cdot f),
\]
where the last equality follows from \cref{clm:number-of-rounds}. Putting these together, we then obtain that 
\[
\Omega\left( j \cdot \frac{b_i^{3/2}}{\alpha(R_{\ell_i,j})}\right) = \sum_{t=1}^j \frac{|R_{\ell_i,t}|^{3/2}}{\alpha(R_{\ell_i,t})} = \widetilde{O}(j \cdot f),
\]
which implies that 
\[
j \cdot f = \widetilde{\Omega}\left( j \cdot \frac{b_i^{3/2}}{\alpha(R_{\ell_i,j})}\right),
\]
and so we can conclude that 
\[
\alpha(R_{\ell_i, j}) = \widetilde{\Omega}\left ( \frac{b_i^{3/2}}{f} \right ).
\]
\end{proof}

We can likewise derive an \emph{upper-bound} on the growth of $\alpha$ by again using our bound on the average progress, but this time in conjunction with our ability to recover an independent set of size $\widetilde{\Omega}\left(\frac{\alpha(S)n}{|S|}\right)$. We provide this derivation below:

\begin{claim}\label{clm:alpha-upper}
    For any $i \in [k]$, and any $j \in [p_i, q_i]$, 
    \[
    \alpha(R_{\ell_i, j}) = \widetilde{O}\left(\frac{j \cdot f \cdot b_i}{n} \right ).
    \]
\end{claim}

\begin{proof}
    As before, because we know that algorithm did not terminate in the $r_{\ell_i, j}$th iteration, we have 
    \[
    \frac{\alpha(R_{\ell_i, j}) \cdot n}{|R_{\ell_i, j}|}\leq r_{\ell_i, j} \cdot f = \widetilde{O}(j \cdot f). 
    \]
    Because 
    \[
     \frac{\alpha(R_{\ell_i, j}) \cdot n}{b_i} \leq\frac{\alpha(R_{\ell_i, j}) \cdot n}{|R_{\ell_i, j}|},
    \]
    this means that 
    \[
    \frac{\alpha(R_{\ell_i, j}) \cdot n}{b_i} =\widetilde{O}(j \cdot f),
    \]
    and so we obtain that 
    \[
    \alpha(R_{\ell_i, j}) = \widetilde{O} \left ( \frac{j \cdot f \cdot b_i}{n} \right ).
    \]
\end{proof}

\paragraph{A Recursive Characterization of $a_i$ and $b_i$ Values.}
We next utilize preceding claims to get the following recursive characterizations of $a_i$ and $b_i$ values.

\begin{lemma} \label{lem:a-b-bound}
Let $i \in [k]$. Then,
\[
a_i = \tilde O\left( \frac{f\cdot b_i}{n}\right), \quad b_i = \tilde O\left( \frac{f^4\cdot a_{i-1}^2}{n^2}\right).
\]
\end{lemma}

\begin{proof}
    By \cref{clm:alpha-upper}, we know that 
    \[
    \alpha(R_{\ell_i, q_i}) = \widetilde{O}\left(\frac{q_i \cdot f \cdot b_i}{n} \right ).
    \]
    Simultaneously, from \cref{lem:alphaGrowth}, we know that 
    \[
    \alpha(R_{\ell_i, q_i}) = \Omega(q_i^2),
    \]
    as $R_{\ell_i, q_i}$ is the $q_i$th set peeled off of size $[2^{\ell_i - 1}, 2^{\ell_i}]$. Together, this implies that 
    \[
    q_i^2 = \widetilde{O}\left(\frac{q_i \cdot f \cdot b_i}{n} \right ),
    \]
    and so it must be the case that 
    \[
    q_i = \widetilde{O}\left(\frac{f \cdot b_i}{n} \right ).
    \]
    Using our relationship between $a_i$ and $q_i$ from \cref{clm:number-of-rounds}, we get that 
    \[
    a_i = \widetilde{O}(q_i) = \widetilde{O}\left(\frac{f \cdot b_i}{n} \right ).
    \]

    To derive the bound on $b_i$, we start by applying the bound of \cref{clm:alpha-upper} to $R_{\ell_i, p_i}$, getting that
    \[
    \alpha(R_{\ell_i, p_i}) = \widetilde{O}\left(\frac{p_i \cdot f \cdot b_i}{n} \right ) = \widetilde{O}\left(\frac{a_{i-1} \cdot f \cdot b_i}{n} \right ),
    \]
    where the final equality uses \cref{clm:number-of-rounds}.

    At the same time, \cref{clm:alpha-lower} ensures that 
    \[
    \alpha(R_{\ell_i, p_i}) = \widetilde{\Omega}\left ( \frac{b_i^{3/2}}{f} \right ).
    \]
    Together then, these bounds imply that 
    \[
    \frac{b_i^{3/2}}{f}  = \widetilde{O} \left ( \frac{a_{i-1} \cdot f \cdot b_i}{n}\right ),
    \]
    and so 
    \[
    b_i^{1/2} = \widetilde{O} \left ( \frac{a_{i-1} \cdot f^2 }{n}\right ),
    \]
    and thus $b_i = \widetilde{O} \left ( \frac{a_{i-1}^2 \cdot f^4 }{n^2}\right )$, as we desire.
\end{proof}

Now that we have established this recursive relationship for the $a_i$'s, we establish the base case value of $a_0$.

\begin{claim}\label{clm:a0-bound}
    $a_0 = \widetilde{O}(f^2 / n)$.
\end{claim}

\begin{proof}
    Recall that we define $a_0 = m'$, where $m'$ is the final good index. So, necessarily among the sets $S_1, \dots S_{a_0}$, at least half of them are good. 

    Because the algorithm does not terminate early, it must be the case that
    \[
    \sum_{j \in [a_0]: S_j \text{ is good}}|S_j| \leq a_0 \cdot f. 
    \]
    In particular, because $a_0$ is a good index, this sum must contain at least $a_0 / 2$ distinct sets. Among these sets, there must be at least $a_0 / 4$ whose size is bounded by $4f$, as otherwise we contradict the stated upper bound on their cumulative size. 

    Now, we let $G = \{S_j: j \in [a_0], |S_j| \leq 4f, S_j \text{ is good} \}$, denote exactly this set of good sets of bounded size. The preceding paragraphs shows that $|G| \geq a_0 / 4$. Next, for $\ell \in [\log(n)]$, we also define $G_{\ell} = \{S \in G: |S| \in [2^{\ell - 1}+1, 2^{\ell}] \}$. By a simple pigeonhole argument, we can easily derive that there exists $\ell^* \in [\log(n)]$ for which 
    \[
    |G_{\ell^*}| \geq \frac{|G|}{\log (n)} \geq \frac{a_0}{4 \log(n)}.
    \]

    Now, let $S$ be the final set in $G_{\ell^*}$. By \cref{lem:alphaGrowth}, we can see that $\alpha(S) = \Omega(|G_{\ell^*}|^2) = \widetilde{\Omega}(a_0^2)$. As $S\in G$, we have $|S|\leq 4f$, and thus
    \[
    \frac{\alpha(S)}{|S|} = \widetilde{\Omega}\left ( \frac{a_0^2}{f}\right).
    \]

    On the other hand, because $S$ appears before $R_{\ell_1, p_1}$, we can again use the monotonicity of the $\frac{\alpha(S_i)}{|S_i|}$ ratios (as per \cref{lem:alphaGrowth}) to see that 
    \[
     \frac{\alpha(S)}{|S|} = O \left ( \frac{\alpha(R_{\ell_1, p_1})}{|R_{\ell_1, p_1}|}\right ) = O \left ( \frac{\alpha(R_{\ell_1, p_1})}{b_1}\right ).
    \]

    Finally, we invoke \cref{clm:number-of-rounds} and \cref{clm:alpha-upper} to see that 
    \[
    \alpha(R_{\ell_1, p_1}) = \widetilde{O} \left ( \frac{p_1  \cdot f \cdot b_1}{n} \right ) = \widetilde{O} \left ( \frac{a_0  \cdot f \cdot b_1}{n} \right ).
    \]
    All together, this means that 
    \[
    \frac{a_0^2}{f} = \widetilde{O} \left ( \frac{a_0  \cdot f \cdot b_1}{n \cdot b_1} \right ),
    \]
    which means that 
    \[
    a_0 = \widetilde{O} \left ( \frac{f^2}{n}\right ),
    \]
    as we desire.
\end{proof}

Now, using this base case established above, along with our recurrence conditions from \cref{lem:a-b-bound}, we can establish the following general bound for the $b_i$'s:

\begin{lemma}\label{lem:b-bound}
    We have 
    \[
    b_1 = \widetilde{O} \left ( \frac{f^8}{n^4}\right ),
    \]
    and more generally, for any $i \geq 1$, we have 
    \[
    b_i \le \frac{f^{7 \cdot 2^i - 6}}{n^{4 \cdot 2^i - 4}} \cdot \log^{O(2^i)}(n).
    \]
\end{lemma}

\begin{proof}
    First, via \cref{lem:a-b-bound} and \cref{clm:a0-bound}, we know that 
    \[
    b_1 = \widetilde{O} \left ( \frac{f^4 \cdot a_0^2}{n^2}\right ) = \widetilde{O} \left ( \frac{f^8}{n^4}\right ).
    \]
    Now, for the general case, by unraveling the $\widetilde{O}$, we know that there exists a  constants $c_1$ such that for large enough $n$, 
    \[
    b_1 \leq \frac{f^8}{n^4} \cdot \log^{c_1}(n).
    \]
    Likewise, for $i \geq 2$, by \cref{lem:a-b-bound}, we know that 
    \[
    b_i = \widetilde{O} \left ( \frac{f^4 \cdot a_{i-1}^2}{n^2}\right )  = \widetilde{O} \left ( \frac{f^6 \cdot b_{i-1}^2}{n^4}\right ),
    \]
    and so there exists a constant $c_2$ such that (for large enough $n$) 
    \[
    b_i \leq \frac{f^6 \cdot b_{i-1}^2}{n^4} \cdot \log^{c_2}(n).
    \]
    Solving the recurrence gives
\[
b_i \leq \frac{f^{7\cdot 2^i-6}}{n^{4\cdot 2^i-4}} \cdot  (\log n)^{2^{i-1}(c_1+c_2)-c_2} = \frac{f^{7\cdot 2^i -6}}{n^{4\cdot 2^i-4}} \log^{O(2^i)} n .
\]
\end{proof}

Next, we make the following simple observation about the number of elements that are contained in bad sets of a given size, across all $m$ sets that are peeled off. 

\begin{claim}
    There exists $\ell^* \in [\log(n)]$ such that 
    \[
    |T_{\ell^*, m}|  = \widetilde{\Omega} \left ( \frac{n}{2^{\ell^*}}\right ).
    \]
\end{claim}

\begin{proof}
    First, we can observe that the total number of elements contained in good sets at the termination of the algorithm is at most $n/4$. Otherwise, if 
    \[
    \sum_{i: S_i \text{ is good}} |S_i| \geq n/4,
    \]
    then this means that we can make average progress per round of $\widetilde{\Omega}(n^{2/3})$, as  we will have deleted $\widetilde{\Omega}(n)$ elements (by definition of them being good), in only $\widetilde{O}(n^{1/3})$ rounds, as per \cref{lem:alphaGrowth}.

    Because we are assuming that we do not hit a break point in \cref{line:return-contract} or \cref{line:return-delete}, we know that the break point instead comes when $|E| < n/2$. Thus, at least $n/2$ elements are removed, and at most $n/4$ of them are contained in good sets. Hence, there are at least $n/4$ elements in bad sets, i.e.,
    \[
    \left | \bigcup_{\ell \in [\log(n)]} \bigcup_{i \in T_{\ell, m}}S_i \right | = \Omega(n).
    \]
     By a pigeonhole argument, we then know that there exists some $\ell^* \in [\log(n)]$ for which 
     \[
     \left | \bigcup_{i \in T_{\ell^*, m}}S_i \right | = \Omega ( n / \log(n)).
     \]
     Because each set $S$ in $T_{\ell^*, m}$ is of size $\leq 2^{\ell^*}$, this immediately means that $T_{\ell^*, m}$ must have $\widetilde{\Omega} \left ( \frac{n}{2^{\ell^*}}\right )$ distinct sets, as we desire.
\end{proof}

\begin{lemma}\label{lem:avgProgNoEarlyStop}
    Let $\ell^* \in [\log(n)]$ be such that $|T_{\ell^*, m}|  = \widetilde{\Omega} \left ( \frac{n}{2^{\ell^*}}\right )$. Then, assuming \cref{alg:new-decomp-2} has not already stopped, then by $r_{\ell^*,|T_{\ell^*,m}|}$-th iteration, the algorithm will have made at least $f = n^{4/7}/\log^c(n)$ average progress per round for sufficiently large constant $c$.
\end{lemma}

\begin{proof}
Let $b = 2^{\ell^*}$. Recall that by the $r_{\ell^*, |T_{\ell^*, m}|}$th iteration, there are two different ways of counting our progress. Either we can contract an independent set, yielding 
\begin{align}
    \Omega \left ( \frac{\alpha(R_{\ell^*, |T_{\ell^*, m}|}) \cdot n}{|R_{\ell^*, |T_{\ell^*, m}|}|}  \right )\label{eq:indsetprogfinal}
\end{align}
total progress (as per \cref{lem:independentProg}),
or we can delete redundant elements, yielding 
\begin{align}
    \widetilde{\Omega} \left ( \sum_{j = 1}^{|T_{\ell^*, m}|} \frac{|R_{\ell^*,j}|^{3/2}}{\alpha(R_{\ell^*,j})}  \right ) \label{eq:deleteredundantfinal}
\end{align}
total progress (as per \cref{thm:balancedProgress}). 

Now, using the fact that the sets $R_{\ell^*,j}$ are all within a constant factor in size, and that the $\alpha$-values are monotonic in such a sequence (\cref{lem:alphaGrowth}), we obtain that 
\[
\sum_{j = 1}^{|T_{\ell^*, m}|} \frac{|R_{\ell^*,j}|^{3/2}}{\alpha(R_{\ell^*,j})}  = \widetilde{\Omega} \left ( |T_{\ell^*, m}| \cdot \frac{b^{3/2}}{\alpha(R_{\ell^*, |T_{\ell^*, m}|})} \right ) = \widetilde{\Omega} \left ( \frac{n}{b} \cdot \frac{b^{3/2}}{\alpha(R_{\ell^*, |T_{\ell^*, m}|})} \right ) = \widetilde{\Omega} \left ( \frac{n \cdot b^{1/2}}{\alpha(R_{\ell^*, |T_{\ell^*, m}|})} \right ),
\]
and we see that 
\[
\frac{\alpha(R_{\ell^*, |T_{\ell^*, m}|}) \cdot n}{|R_{\ell^*, |T_{\ell^*, m}|}|}  = \Theta \left ( \frac{\alpha(R_{\ell^*, |T_{\ell^*, m}|}) \cdot n}{b} \right ).
\]

Next, we see that 
\[
\max \left ( \frac{\alpha(R_{\ell^*, |T_{\ell^*, m}|}) \cdot n}{b} , \frac{n \cdot b^{1/2}}{\alpha(R_{\ell^*, |T_{\ell^*, m}|})}\right ) \geq \frac{n}{b^{1/4}},
\]
as the worst case is when 
\[
\frac{\alpha(R_{\ell^*, |T_{\ell^*, m}|}) \cdot n}{b} = \frac{n \cdot b^{1/2}}{\alpha(R_{\ell^*, |T_{\ell^*, m}|})},
\]
which happens when $\alpha(R_{\ell^*, |T_{\ell^*, m}|}) = b^{3/4}$.

Importantly then, one of the above methods for making progress always yields $\widetilde{\Omega} \left ( \frac{n}{b^{1/4}}\right ) =\widetilde{\Omega} \left ( n^{3/4}\right ) $ total progress.

Now, to bound our average progress, we upper bound the number of rounds until this point, i.e., we bound $r_{\ell^*, |T_{\ell^*, m}|}$. For this, by \cref{lem:a-b-bound} and \cref{lem:b-bound}, we know that for some value of $i$, 
\[
r_{\ell^*, |T_{\ell^*, m}|} \leq a_i = \widetilde{O} \left ( \frac{f \cdot b_i}{n}\right ) = \frac{f^{7 \cdot 2^i - 5}}{n^{4 \cdot 2^i - 3}} \cdot \log^{O(2^i)}(n).
\]
In particular, as long as we choose $f = \frac{n^{4/7}}{\log^{c}(n)}$ for a sufficiently large constant $c$, this satisfies 
\[
r_{\ell^*, |T_{\ell^*, m}|} \leq \frac{n^{4 \cdot 2^i - 20/7}}{n^{4 \cdot 2^i - 3}} \cdot \frac{\log^{O(2^i)}(n)}{\log^{7 \cdot 2^i \cdot c - 5c}(n)}  \leq n^{1/7}.
\]

To conclude then, this means that the average progress made via either \cref{eq:indsetprogfinal} or \cref{eq:deleteredundantfinal} is 
\[
\widetilde{\Omega} \left ( \frac{n^{3/4}}{r_{\ell^*, |T_{\ell^*, m}|}}\right ) = \widetilde{\Omega} \left ( \frac{n^{3/4}}{n^{1/7}}\right ) = \widetilde{\Omega} \left ( n^{4/7}\right ) \geq f.
\]
\end{proof}

\paragraph{Completing the Proof of Main Theorem.}
With this, we re now ready to provide a proof of our main theorem.

\begin{proof}[Proof of \cref{thm:3-7rounds}]
    To start, we set $f = \frac{n^{4/7}}{\log^{c}(n)}$ for a sufficiently large constant $c$.
    
    First, we can observe that if the algorithm ever terminates in \cref{line:return-contract} during iteration $i$, then by \cref{lem:independentProg}, we are recovering an independent set of size $\Omega \left ( \frac{\alpha(S_i)}{|S_i|} \cdot \frac{n}{2}\right ) = \Omega\left(\frac{\widehat \alpha(S_i)}{|S_i|} n\right) = \Omega( i \cdot f)$, where we have used the fact here that $|E| \geq n/2$. This ensures that we make average progress $\Omega(f) = \widetilde\Omega(n^{4/7})$.

    Likewise, we can observe that if the algorithm ever terminates in \cref{line:return-delete} during some iteration $i$, then the value of $T$ during iteration $i$ is $\Omega(i \cdot f)$. In particular, 
    \[
    T = \sum_{j \in [i]: S_j \text{ is good}} |S_j| + \sum_{j \in [i]: S_j \text{ is bad}} \frac{|S_j|^{3/2}}{\widehat \alpha(S_j)}.
    \]
    By \cref{thm:balancedProgress}, we know that for any good set, we can delete $\widetilde{\Omega}(|S_j|)$ redundant elements, and for any bad set, we can delete $\widetilde{\Omega}\left ( \frac{|S_j|^{3/2}}{\alpha(S_j)}\right ) = \widetilde{\Omega}\left ( \frac{|S_j|^{3/2}}{\widehat\alpha(S_j)}\right )$ redundant elements. Thus, using \cref{thm:balancedProgress}, we can indeed delete $\widetilde{\Omega}(T)$ elements, thus yielding 
    \[
     \widetilde{\Omega} \left ( \frac{T}{i}\right ) = \widetilde{\Omega}\left (\frac{i \cdot f}{i} \right) =  \widetilde \Omega (f)
    \]
    average progress per round. 

    Otherwise, suppose we do not terminate on \cref{line:return-contract}  or \cref{line:return-delete}, then as we chose $c$ to be a sufficiently large constant, we can invoke \cref{lem:avgProgNoEarlyStop}, with the value of $f = \frac{n^{4/7}}{\log^{c}(n)}$. This again guarantees that there is some iteration $i$ for which our value of $T$ exceeds $\Omega \left ( i\cdot f \right )$, giving a contradiction, as this implies that we do indeed terminate in \cref{line:return-contract} or \cref{line:return-delete}.

    In order to get our desired probability bound, observe that every invocation of \cref{alg:globalOptimal} succeeds in finding a globally optimal set with probability $1 - 2^{- \Omega(n)}$. Additionally, each invocation of \cref{thm:balancedProgress} succeeds with probability $1 - 2^{- \Omega(n)}$.

    Finally, this means that in each invocation of \cref{alg:new-decomp-2}, we invest $i$ rounds, and either contract on an independent set of size $\widetilde{\Omega}(i \cdot n^{4/7})$, or delete $\widetilde{\Omega}(i \cdot n^{4/7})$ redundant elements (each invocation succeeds with probability $1 - 2^{- \Omega(n)}$). If we let $F(n)$ denote the number of rounds required to find a basis of a matroid on $n$ elements, we get the recurrence that 
    \[
    F(n) = i + F(n - \widetilde{\Omega}(i \cdot n^{4/7})),
    \]
    which implies that $F(n) = \widetilde{O}(n^{3/7})$, as we desire. To avoid issues with our probability bound when the number of elements in the matroid becomes small, observe that once $|E| = n^{1/2}$, we can simply invoke \cite{KUW85} and deterministically find a basis of the remaining elements, in $O(n^{1/4})$ additional rounds. Thus, all together, $F(n) = \widetilde{O}(n^{3/7})$, and our probability of success is $1 - 2^{- \Omega(\sqrt{n})}$.
\end{proof}

\section{Acknowledgments}

The authors would like to thank Eric Balkanski for pointing out the applications to faster optimization under matroid constraints.

\bibliographystyle{alpha}
\bibliography{ref}

\appendix

\section{Missing Proofs in \cref{sec:intro}}\label{sec:matroidIntersection}

\subsection{Proof of \cref{thm:matroid-intersection-intro}}
To start, we recall the following lemma of \cite{BT25}:

\begin{lemma}[\cite{BT25}]
Let $\M_1=(E,\I_1),\M_2=(E,\I_2)$ be 2 matroids. Let $n=|E|$ and $r$ be the size of the largest independent set of $\M_1,\M_2$. Then
\begin{itemize}
    \item There is an $\widetilde O\left(\frac{nT(n)}{\varepsilon\Delta}\right)$ rounds independence-query algorithm that finds a common independent set $S\in I_1\cap I_2$ of size $|S|\geq r-(\varepsilon r+\Delta)$, given that there is a $T(n)$ round independence-query algorithm that finds a basis of a matroid on $n$ elements.
    \item Given $S\in I_1 \cap I_2$, in a single round of independence query, one can compute an $S'\in I_1 \cap I_2$ of size $|S'|=|S|+1$ or decide that $S$ is of maximum possible size.
\end{itemize}
\end{lemma}

Now, we show \cref{thm:matroid-intersection-intro}:

\begin{theorem}[\cref{thm:matroid-intersection-intro} restated]
For two matroids $\M_1,\M_2$ on $n$ elements, there is an $\widetilde{O}(n^{17/21})$ round algorithm, making polynomially many independence queries per round, which recovers a maximum common independent set of $\M_1$ and $\M_2$ with high probability. 
\end{theorem}
\begin{proof}
We set $\varepsilon = n^{10/21}r^{-2/3}$ and $\Delta = \varepsilon r= n^{10/21}r^{1/3}$. We can first find an $S\in I_1\cap I_2$ of size $|S|\geq r-(\varepsilon r +\Delta)$ in
\[
\widetilde O\left( \frac{n\cdot n^{3/7}}{\varepsilon \Delta} \right) = \widetilde O(n^{10/21} r^{1/3})
\]
rounds and then augment it to optimal in $O(\varepsilon r+\Delta)=O(n^{10/21}r^{1/3})$ rounds. As $r\leq n$, the total rounds of adaptivity needed is $\widetilde O(n^{17/21})$.
\end{proof}

\subsection{Proof of \cref{thm:submodular-maximization}}

\begin{definition}[\cite{BRS19b}]
Given a matroid $\M=(E,\I)$, we say that $(a_1,\dots,a_{\rank(\M)})$ is a \emph{random feasible sequence} if for all $i\in [\rank(\M)]$, $a_i$ is an element chosen uniformly at random from $\{a\in E\setminus \{a_1,\dots,a_{i-1}\}:\{a_1,\dots,a_{i-1},a\}\in \I\}$.
\end{definition}

\begin{lemma}\label{lem:generate-random-sequence}
Let $\M=(E,\I)$ be a matroid on $n$ elements. There is an $\tilde O(n^{3/7})$ round algorithm that generates a random feasible sequence in $\M$.
\end{lemma}
\begin{proof}
We follow Algorithm 4 of \cite{BRS19b}. First, we sample a random permutation $\{e_1,\dots,e_n\}$ of the ground set $E$, and compute $r_i=\rank(\{e_1,\dots,e_i\})$ for every $i\in [n]$. This step takes $\tilde O(n^{3/7})$ rounds using $n$ parallel calls to our basis-finding algorithm from \cref{thm:3-7rounds}. Then for every $i\in [\rank(\M)]$, we let $a_i$ be the $i$th element $e_j$ such that $r_j-r_{j-1}=1$. The resulting sequence $(a_1,\dots,a_{\rank(\M)})$ is a random feasible sequence. See Lemma 15 of \cite{BRS19b} for a proof.
\end{proof}

\begin{lemma}[\cite{BRS19b}]\label{lem:random-sequence}
Let $\M$ be a matroid on $n$ elements. Given that there is a $T(n)$ round independence-query algorithm that generates a random feasible sequence in $\M$, for every $\epsilon>0$, there is an $\tilde O(T(n)\cdot \epsilon ^{-3})$ round algorithm that obtains an $1-1/e-O(\epsilon)$ approximation for maximizing a monotone submodular function under the matroid constraint $\M$ with high probability.
\end{lemma}

\begin{theorem}[\cref{thm:submodular-maximization} restated]
For any $\epsilon > 0$, there is an $\widetilde{O}(n^{3/7}\epsilon^{-3})$-round algorithm that, given any matroid $\M=(E,\I)$ on $n$ elements and any monotone submodular function $f : 2^E \to \R_{\ge 0}$, 
makes polynomially many independence queries per round and outputs, with high probability, a $(1 - 1/e - O(\epsilon))$-approximation to the maximum of $f$ under the matroid constraint~$\M$.
\end{theorem}
\begin{proof}
It follows immediately from \cref{lem:generate-random-sequence} and \cref{lem:random-sequence}.
\end{proof}

\section{Proof of \cref{lem:alphaGrowth}} \label{sec:lem-proof}

\newcommand{\Hyp}{\text{Hyp}}
\newcommand{\Var}{\textbf{Var}}

We only prove the last claim of \cref{lem:alphaGrowth} here as the rest is given by Theorem 4.16 in \cite{KPS25}.

\begin{lemma}
Let $\mathcal{M}$ be a matroid, and let $S_1, \dots S_k$ be a sequence of sets that are peeled off in accordance with \cref{alg:iterativePeel}. Then, for any $i < j \in [k]$,
\[
\frac{\alpha(S_j)}{|S_j|} = \Omega \left(\frac{\alpha(S_i)}{|S_i|}\right).
\]
\end{lemma}

\begin{proof}
Let $\ell = \frac{\alpha(S_i)}{10|S_i|} n$. Suppose $U$ is a subset of the ground set $E$ of cardinality $\ell$, we define random variables $X_i = |U\cap S_i|, X_j =|U \cap S_j|$. Note that $X_i\sim \Hyp(n,|S_i|,\ell),X_j\sim \Hyp(n,|S_j|,\ell)$ and they are negatively correlated.

For the sake of contradiction, suppose
\[
\frac{\alpha(S_j)}{|S_j|} \leq \frac{1}{100} \cdot\frac{\alpha(S_i)}{|S_i|}
\]
We aim to show that
\[
\Pr[X_i< \alpha(S_i) \land X_j\geq \alpha(S_j)] \geq 1/2.
\]
Consider the $i$-th iteration of \cref{alg:iterativePeel}, the above implies that for a random permutation $\pi$ of $[n]$, there is a $\geq 1/2$ probability that there are less than $\alpha(S_i)$ elements from $S_i$ and more than $\alpha(S_j)$ elements from $S_j$ in the first $\ell$ elements of $\pi$. Conditioned on this, with at least $1/4$ probability, there is no circuit within $S_i$ but there is a circuit in $S_j$. This implies that with $\geq 1/8$ probability, the first circuit appears outside of $S_i$ when adding elements according to the order of a random permutation $\pi$. But on the other hand, since $S_i$ is a greedility-optimal set in the $i$-th iteration, by \cref{clm:boundedErrorEstimation}, we have $q(S) \geq 1-2^{-20}-1/n^2\geq 1/2$. This is a contradiction.

Since we have
\[
\E[X_i] = \frac{\alpha(S_i)}{10},
\]
by Markov's inequality,
\[
\Pr[X_i < \alpha(S_i)]  = 1-\Pr[X_i\geq \alpha(S_i)] = 1- \Pr[X_i\geq 10\E[X_i]] \geq \frac{9}{10}.
\]
On the other hand, we have
\[
\E[X_j] = \ell\frac{|S_j|}{n} = \frac{\alpha(S_i)|S_j|}{10|S_i|} \geq 10\alpha(S_j), \quad \Var[X_j] = \ell \frac{|S_j|}{n}\left(1-
\frac{|S_j|}{n}\right) \leq \E[X_j].
\]
By Chebyshev's inequality,
\begin{align*}
\Pr[X_i\geq \alpha(S_j)] & = 1-\Pr[X_j<\alpha(S_j)] \\
& \geq 1-\Pr[X_j\leq 0.1 \cdot \E[X_j]] \\
& \geq 1- \Pr[|X_j-\E[X_j]|\geq 0.9 \cdot \E[X_j]]\\
& \geq 1- \frac{\Var[X_j]}{0.81 \cdot \E[X_j]^2}\\
& \geq 1- \frac{1}{0.81\cdot\E[X_j]}\\
& \geq 1-\frac{1}{8.1 \cdot \alpha(S_j)} \geq \frac{7}{8}.
\end{align*}

Since $X_i,X_j$ are negatively correlated, we have
\[
\Pr[X_i<\alpha(S_i)\land X_j\geq \alpha(S_j)] \geq \Pr[X_i<\alpha(S_i)] \cdot \Pr[X_j\geq \alpha(S_j)] \geq \frac{9}{10}\cdot \frac{7}{8} \geq \frac{1}{2}.
\]
as desired.
\end{proof}

\section{A Proof of \cref{clm:basicRecoverRedundant}}\label{sec:appendixS2alpha2}

In this section, we provide a proof of \cref{clm:basicRecoverRedundant}. Note that this proof exactly follows \cite{KPS25}'s proof, but neglects to make some simplifications. We include a restatement of the claim below:

\begin{claim}[Restatement of \cref{clm:basicRecoverRedundant}.]
    Let $\M=(E,\I)$ be a matroid and $S\subseteq E$, such that for every $x \in S$, $p_{x, \M|_S} \geq 1 / |S|^2$. Then, there is a one round algorithm making polynomially many queries which recovers 
\[
\widetilde{\Omega} \left ( \sum_{x \in S} \min \left (1, p_{x, \M|_S} \cdot \frac{|S|^2}{\alpha(S)^2} \right ) \right )
\]
redundant elements in $S$ with probability $1 - 2^{- \Omega(n)}$.
\end{claim}

Before providing the proof, we also recall some bounds on when dependences arise in relation to the $\alpha$ value of $S$:

\begin{claim} \label{clm:alpha-dependent}[Claim 4.7 of \cite{KPS25}]
For any set $S$ and any integer $d>1$, a random subset of $S$ of cardinality $d\cdot\alpha(S)$ is dependent with probability at least $1-2^{-d}$.
\end{claim}

\begin{proof}[Proof of \cref{clm:basicRecoverRedundant}.]
    To start, we recall the algorithm used in \cite{KPS25}:

    \begin{algorithm}[H]\caption{RecoverRedundantElements$(S)$}\label{alg:recover-reduant-elements}
    $t\gets 20\log n\cdot \alpha(S), \ell \gets \frac{|S|}{4t}$\\
    \For{$i \in [\ell]$ in parallel} {
        Draw a random permutation (bijection) $\pi:[|S|]\to S$.\\
        $A_i \gets \{\pi(1),\dots, \pi(t)\},B_i\gets \emptyset$.\\
        \For{$j \in [t]$ in parallel}{
            \For{$x\in E\setminus A_i$ in parallel} {
                Query $\mathrm{Ind}(\{\pi(1),\dots,\pi(j)\})$ and $\mathrm{Ind}(\{\pi(1),\dots,\pi(j)\}\cup \{x\})$.\\
                \If{$\mathrm{Ind}(\{\pi(1),\dots,\pi(j)\})=1\land\mathrm{Ind}(\{\pi(1),\dots,\pi(j)\}\cup \{x\}) =0$} {
                    $B_i\gets B_i\cup \{x\}.$
                }
            }
        }
    }
    \Return{$\bigcup_{i\in[\ell]} B_i \setminus (\bigcup_{i\in[\ell]} A_i)$}
\end{algorithm}

First, we note that
\[
\left(\bigcup_{i\in[\ell]} B_i \setminus \left(\bigcup_{i\in[\ell]} A_i\right)\right) \subseteq \text{span}\left(\bigcup_{i\in[\ell]} A_i\right)
\]
Thus, the recovered set is indeed redundant. We now focus on bounding the size of this set in the following.

Now, for every $x\in S$, we say that $p_{x,r}$ is the probability over a random order of picking elements of $S$ such that $x$ is the $r$-th element added, that $x$ participates in the first circuit that appears. Note that in this proof, we always use $p_x$ to mean $p_{x, \M|_S}$; that is to say, everything is being done with respect to the subset $S$ of elements. 

In particular, we can establish some simple inequalities:
\begin{enumerate}
    \item $p_{x,r} \geq p_{x,r + 1}$.
    \item $p_{x, \M|_S} = \frac{1}{|S|} \sum_{r = 1}^{|S|} p_{x,r}$.
    \item $\sum_{r = t+1}^{|S|} p_{x,r} \leq 1 / n^{10}$.
\end{enumerate}
The third inequality follows from \cref{clm:alpha-dependent}: a random subset of $S$ of cardinality $t=20\widehat\alpha(S)\log n\geq 10\alpha(S)\log n$ is dependent with probability at least $1-1/n^{10}$.

With this, we can see that for any $x\in S$,
\[
    p_x =\frac{1}{|S|} \sum_{r = 1}^{|S|} p_{x,r} = \frac{1}{|S|}\left( \sum_{r = 1}^{t} p_{x,r} + \frac{1}{n}\right) \leq \frac{t}{|S|} \cdot p_{x,1} + \frac{1}{n^{10}}.
\]
In particular, this implies that 
\[
    p_{x,1} \geq \frac{|S|}{t} \cdot \left(p_x - \frac{1}{n^{10}}\right) \geq \frac{|S|}{t} \cdot (p_x)/2,
\]
as we assumed that $p_{x, \M|_S} \geq 1 / |S|^2 \geq 1 / n^2$.

Now, let us revisit the above algorithm. Our first step will be to understand the probability that an element $x$ appears in one of the sets $A_1, \dots A_{\ell}$. For this, observe that each set $A_i$ is of size $t$. Thus, 
\[
    \Pr\left[x \notin \bigcup_{i\in[\ell]} A_i \right] = \Pr[x \notin A_1]^{\ell} = \left ( 1 - \frac{t}{|S|}\right )^{\ell} =\left ( 1 - \frac{t}{|S|}\right )^{\frac{|S|}{4t}} \geq  e^{-1/2} \geq 1/2.
\]
The first inequality is because $\frac{t}{|S|} = \frac{20 \log n \cdot \widehat \alpha(S)}{|S|}\leq \frac{40\log n\cdot \alpha(S)}{|S|} \leq \frac{1}{2}$ and for every $0\leq x\leq \frac{1}{2}, 1-x\geq e^{-2x}$.

Now, let us introduce the value $q_i$ such that 
\[
    q_x = \Pr\left[x \notin \bigcup_{i\in[\ell]} A_i \wedge x \in \bigcup_{i\in[\ell]} B_i\right] = \Pr\left[x \notin \bigcup_{i\in[\ell]} A_i\right] \cdot \Pr\left[x \in \bigcup_{i\in[\ell]} B_i\;\middle|\; x \notin \bigcup_{i\in[\ell]} A_i\right].
\]
Note that the samples $A_1, \dots A_{\ell}$ are all done independently of one another. Hence,
\[
\Pr\left[x \in \bigcup_{i\in[\ell]} B_i\;\middle|\; x \notin \bigcup_{i\in[\ell]} A_i\right] = 1 - \Pr\left[x \notin \bigcup_{i\in[\ell]} B_i\;\middle|\; x \notin \bigcup_{i\in[\ell]} A_i\right] = 1-\Pr[x\not\in B_1\mid x\notin A_1]^{\ell}.
\]
Now, let us understand $\Pr[x\in B_1 | x\notin A_1]$. This is exactly the probability of $x$ appearing in the first circuit when we randomly add the set $A_1$ of elements to $x$. Since $A_1$ is disjoint from $x$, this is exactly $p_{x,1}$. Hence, we obtain that
\begin{align*}
1 - \Pr[x \notin B_1 | x \notin A_1]^{\ell} & = 1 - (1 - p_{x,1})^{\ell}\\
& \geq 1 - \left ( 1 - \frac{|S| \cdot p_{x, \M|_S}}{2t} \cdot \right)^{\frac{|S|}{4t}} \\
& \geq 1-\exp\left(\frac{|S|^2 \cdot p_{x, \M|_S}}{8t^2\log n}\right)\\
& \geq \min\left\{\frac{1}{2},\frac{|S|^2 \cdot p_{x, \M|_S}}{16t^2\log n}\right\}
\end{align*}
The last inequality follows from the fact that $1-e^{-x}\geq \min\{1/2,x/2\}$ when $x\geq 0$.

To conclude, we obtain that 
\[
q_x \geq \frac{1}{2} \cdot (1 - \Pr[x \notin B_1 | x \notin A_1]^{\ell}) \geq \min\left \{\frac{1}{4}, \frac{|S|^2 \cdot p_{x, \M|_S}}{32t^2\log n} \right \}
\]
Finally then, we see that 
\begin{align*}
\E\left[\left|\bigcup_{i\in[\ell]} B_i \setminus \left(\bigcup_{i\in[\ell]} A_i\right)\right|\right] & = \sum_{x\in S} q_x\\
& \geq \sum_{x\in S}\min\left \{\frac{1}{4}, \frac{|S|^2 \cdot p_{x, \M|_S}}{32t^2\log n} \right \}\\
& = \widetilde{ \Omega} \left ( \sum_{x\in S}\min\left \{1, \frac{|S|^2 \cdot p_{x, \M|_S}}{\alpha(S)^2} \right \}\right ). 
\end{align*}
Repeating the above $\text{poly}(n_0)$ times achieves at least this expectation with probability at least $1 - 2^{-n_0}$ by an invocation of Hoeffding's inequality. This concludes the claim. 
\end{proof}

\section{An $\widetilde{O}(n^{4/9})$ Round Algorithm}\label{sec:appendix49}

In this section, we show how, with careful analysis, just using the subroutines for making progress from \cite{KPS25} enables an algorithm which requires $\widetilde{O}(n^{4/9})$ rounds for finding a basis. Recall that \cite{KPS25} had two primary ways of making progress on the sets $S_i$ that are peeled off during the decomposition:
\begin{enumerate}
	\item A $1$ round algorithm for finding $\widetilde{\Omega} \left (\frac{\alpha(S_i)}{|S_i|} \cdot n \right )$ independent elements.
	\item A $1$ round algorithm for finding $\widetilde{\Omega} \left ( \frac{|S_i^2|}{\alpha(S_i)^2}\right )$ redundant elements inside $S_i$. 
	 \end{enumerate}
	 
Additionally, we take advantage of their characterization of the \emph{growth} of $\alpha$ values as we continue to peel off sets (see \cref{lem:alphaGrowth}). With the setting established, we now fix our target progress per round as $t=n^{5/9}$ and modify the decomposition algorithm as in \cref{alg:new-decomp} (see below).

\begin{algorithm}[H]
\caption{NewDecomposition$(\M)$}
\label{alg:new-decomp}
    $\M\gets \text{RemoveSmallCircuits}(\M)$\\
    $T\gets 0$\\
    $i\gets 0$\\
    \While{$|E|\geq n/2$} {
        $i\gets i+1$\\
        $S_i\gets \text{FindGreedilyOptimal}(\M)$ \tcp{Algorithm 6 of \cite{KPS25}}
        Let $\widehat \alpha(S_i)$ be the estimation of $\alpha(S_i)$, which satisfies $(\alpha(S_i)-1)/2\leq \widehat \alpha(S_i)\leq 2\alpha(S_i)$ with high probability. \tcp{see Definition 4.6 and Claim 4.8 of \cite{KPS25}}
        \If{$\widehat \alpha(S_i)\leq |S_i|^{1/2}$}  {
            $T\gets T+|S_i|$\\
            \If{$T\geq i\cdot t$} {
                \Return \label{line:return-1}
            }
        }
        \Else {
            \Return $S_i$ \label{line:return-2}
        }
        $\M\gets \M\setminus S_i$
    }
\end{algorithm}

\begin{claim} \label{clm:number-small}
Suppose \cref{alg:new-decomp} does not return on \cref{line:return-1}. Then, before the first $S_i$ with $\alpha(S_i)>|S_i|^{1/2}$, there are at most $O(\sqrt{t\alpha(S_i)/|S_i|})$ rounds.
\end{claim}
\begin{proof}
As the algorithm does not return on \cref{line:return-1}, we have that
\[
\sum_{j\leq i} |S_j| \leq i\cdot t,
\]
which implies at most half of $S_j$'s has $|S_j|> 2t$, which means
\[
i \leq 2\sum_{j\leq i} \textbf{1}[|S_j|\leq 2t].
\]

By \cref{lem:alphaGrowth}, for any $\ell\in [\log n]$, we have
\[
|\{j<i:S_j \in [2^\ell,2^{\ell+1}]\}| \leq \sqrt{\frac{\alpha(S_i)}{|S_i|} \cdot 2^{\ell+1}}.
\]
Therefore,
\[
i\leq 2\sum_{\ell=0}^{\log t+1} |\{j<i:S_j \in [2^\ell,2^{\ell+1}]\}| \leq 2 \sqrt{\frac{\alpha(S_i)}{|S_i|} \cdot 2^{\ell+1}} \leq 2\sqrt{\frac{\alpha(S_i)}{|S_i|} \cdot t}.
\]
\end{proof}

\begin{lemma} \label{lem:old-progress}
\cref{alg:new-decomp} makes $\tilde\Omega(n^{5/9})$ progress per round on average.
\end{lemma}
\begin{proof}
When \cref{alg:new-decomp} peels off an greedily optimal set $S$ (see Definition 4.9 in \cite{KPS25}), in one additional round, we can either find an independent set of size $\frac{\alpha(S)}{|S|} n$ by \cref{lem:independentProg}, or recover
\[
\widetilde{\Omega} \left ( \sum_{i \in S} \min \left (1, p_{i, \M|_S} \cdot \frac{|S|^2}{\alpha(S)^2} \right ) \right )
\]
redundant elements by \cref{clm:basicRecoverRedundant}. Since $S$ is a greedily optimal set, we have $p_{i,\M|_S}=\tilde \Omega(\frac{1}{|S|})$ for every $i\in S$ and thus
\[
\widetilde{\Omega} \left ( \sum_{i \in S} \min \left \{1, p_{i, \M|_S} \cdot \frac{|S|^2}{\alpha(S)^2} \right \} \right ) = \widetilde \Omega \left(\min\left\{|S|,\frac{|S|^2}{\alpha(S)^2}\right\}\right).
\]
Therefore, if $\alpha(S_i)=O( |S_i|^{1/2})$, we can recover $\widetilde{\Omega} (|S|)$ redundant elements.

If \cref{alg:new-decomp} returns on \cref{line:return-1}, as every $S_j$ that the algorithm recovers satisfies $\alpha(S_j) = \Theta(\widehat \alpha(S_j)) = O( |S_j|^{1/2})$, we can delete $\sum_{j\leq i} \widetilde \Omega(|S_j|)$ elements as shown in the preceding paragraph. As $T=\sum_{j\leq i} |S_j|$ and $T\geq i\cdot t$, we see that the average progress the algorithm makes is
\[
\frac{\sum_{j\leq i} \widetilde \Omega(|S_j|)}{i} = \frac{\widetilde \Omega(T)}{i} = \widetilde{\Omega}(t) = \widetilde{\Omega}(n^{5/9}).
\]

Otherwise, if the algorithm returns on \cref{line:return-2}, by \cref{clm:number-small}, there are at most $O(\sqrt{t\alpha(S_i)/|S_i|})$ rounds before the algorithm finds the first $S_i$ with $\widehat\alpha(S_i)>|S_i|^{1/2}$. After finding such $S_i$, as shown in the first paragraph, we can make
\[
\widetilde \Omega \left(\max\left\{\frac{\alpha(S_i)}{|S_i|}n,\frac{|S_i|^2}{\alpha(S_i)^2}\right\}\right)
\]
progress by either contracting an independent set or removing redundant elements. Therefore the average progress per round is at least
\[
\widetilde \Omega \left( \frac{\max\left\{\frac{\alpha(S_i)}{|S_i|}n,\frac{|S_i|^2}{\alpha(S_i)^2}\right\}}{\sqrt{\frac{\alpha(S_i)}{|S_i|}t}} \right)\geq \widetilde \Omega\left(\frac{n^{5/6}}{\sqrt{t}}\right) = \widetilde \Omega (n^{5/9}).
\]

Suppose the algorithm does not hit a break point in both \cref{line:return-1} and \cref{line:return-2}, then we know that $\sum_{j\leq i}|S_j|\geq n/2$. As all these $S_j$'s satisfies $\alpha(S_j)\geq |S_j|^{1/2}$ (otherwise the algorithm should have return on \cref{line:return-2}), we have that $T=\sum_{j\leq i}|S_j|\geq n/2$. By \cref{lem:alphaGrowth}, $i=O(n^{1/3})$, so $i\cdot t = O(n^{1/3}\cdot n^{5/9}) = O(n^{8/9})\leq T$, which implies that the algorithm should return on \cref{line:return-1}, giving a contradiction.
\end{proof}

\begin{theorem}
For an arbitrary matroid $\M$ on $n$ elements, there is a $\widetilde{O}(n^{4/9})$ round algorithm, making polynomially many independence queries per round, which recovers a basis of $\M$ with high probability.
\end{theorem}

\begin{proof}
In each invocation of \cref{alg:new-decomp}, we invest $i$ rounds, and by \cref{lem:old-progress}, we either contract an independent set of size $\widetilde \Omega(i\cdot n^{5/9})$, or delete $\widetilde \Omega(i\cdot n^{5/9})$ redundant elements (each invocation succeeds with probability $1-2^{-\Omega(n)}$). If we let $F(n)$ denote the number of rounds required to find a basis of a matroid on $n$ elements, we get the recurrence that 
    \[
    F(n) = i + F(n - \widetilde{\Omega}(i \cdot n^{5/9})),
    \]
    which implies that $F(n) = \widetilde{O}(n^{4/9})$, as we desire.
\end{proof}

\end{document}